\documentclass[envcountsame]{llncs}

\usepackage{mymacros}
\usepackage{url} 
\usepackage[ruled,lined]{algorithm2e}
\usepackage{graphicx}
\usepackage{amsmath}

\usepackage{verbatim}
\usepackage{color}
\usepackage{pdfsync} 
\usepackage{amsfonts}
\usepackage[english]{babel} 

 \usepackage{colortbl} 
 \definecolor{Gray}{gray}{0.8}
 
\usepackage{multicol}
 \usepackage{multirow}

\begin{document}


\title{ Modelling cooperating failure-resilient Processes}
\author{R\"udiger Valk}
\institute{University of Hamburg, Department of Informatics\\
Hamburg, Germany \\
  \email{ruediger.valk@uni-hamburg.de}}
\maketitle

\begin{abstract}
Cycloids are particular Petri nets for modelling processes of actions or events. They belong to the fundaments of Petri's general systems theory and have very different interpretations, ranging from Einstein's relativity theory and  elementary information processing gates to the modelling of interacting sequential processes. 
The subclass of regular cycloids describes cooperating sequential processes. 
Such cycloids are extended to cover failure resilience.\end{abstract}



\section{Introduction}\label{sec-intro}
Cycloids have been introduced  by C.A. Petri in \cite{Petri-NTS} in the section on physical spaces, using as examples firemen carrying the buckets with water to extinguish a fire, the shift from Galilei to Lorentz transformation and the representation of elementary logical gates like Quine-transfers.
Based on  formal descriptions of cycloids in \cite{Kummer-Stehr-1997} and \cite{fenske-da} a more elaborate formalization is given in \cite{Valk-2019}.
\begin{figure}[htbp]
 \begin{center}
        \includegraphics [scale = 0.31]{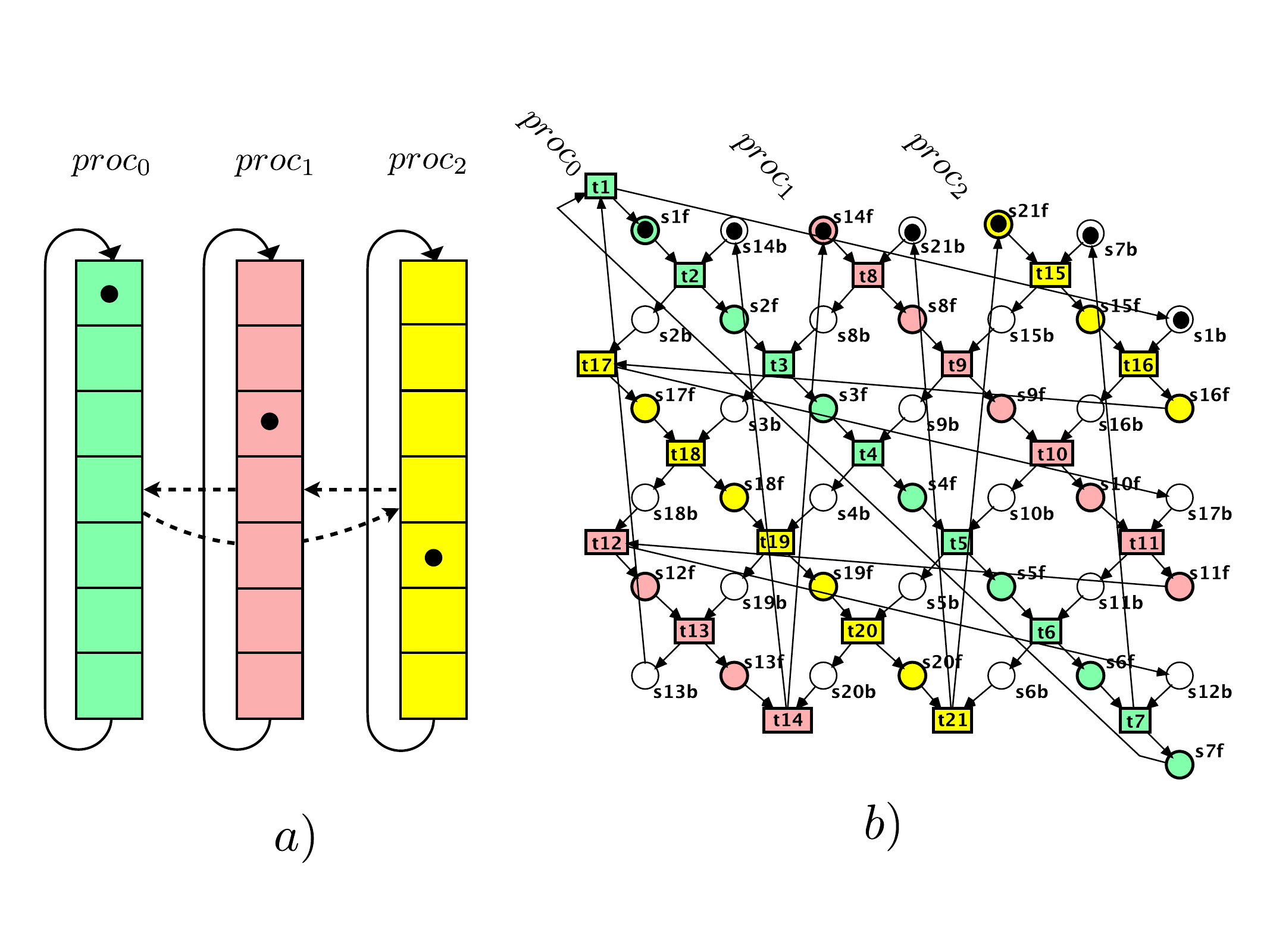}
        \caption{Three sequential processes synchronized by single-bit channels, }
        \label{c-4-3-3-3-b}
      \end{center}
\end{figure}

To give an application for the theory, as presented in this article,  consider a distributed system of a finite number of circular and sequential processes. The processes are synchronized by uni-directional one-bit channels in such a way that they behave like a circular traffic queue when folded together. To give an example, Figure \ref{c-4-3-3-3-b}a) shows three such sequential circular processes, each of length $7$. In the initial state the control is in position $1$, $3$ and $5$, respectively. The synchronization, realized by the connecting channels, should be such as the three processes would be folded together. This means, that the controls of $proc_0$ and $proc_1$ can make only one  step until the next process $proc_2$ makes a step itself, while
the control of $proc_2$ can make  two steps until $proc_0$ makes a step. Following  \cite{Valk-2020} this behaviour is realized by the cycloid of 
Figure \ref{c-4-3-3-3-b}b) modelling the three processes 
by the transition sequences $proc_0 =$ [\textbf{t1 t2 $\cdots$ t7}], as well as $proc_1 =$ [\textbf{t8 t9 $\cdots$ t14}] and $proc_2 =$ [ \textbf{t15 t16 $\cdots$ t21}]. The channels are represented by the safe places connecting these processes. By this example the power of the presented theory is shown, since the rather complex net is  unambiguously determined by the parameters 
 $\mathcal{C} ( \alpha, \beta, \gamma, \delta ) = \mathcal{C}(4,3,3,3)$.

We recall some standard notations for set theoretical relations.
If $R\subseteq A \!\times\! B$ is a relation and $U \subseteq A$ then 
$R[U]:= \{b\,|\,\exists u \in U: (u,b)\in R\}$ is the \emph{image} of $U$ and $R[a]$ stands for $R[\{a\}]$. 
$R^{-1}$ is the \emph{inverse relation} and $R^+$ is the \emph{transitive closure} of $R$ if $A=B$.
Also, if $R\subseteq A \!\times\! A$ is an equivalence relation then $\eqcl[R]{a}$
 is the \emph{equivalence class} of the quotient $A/R$ containing $a$.
Furthermore  $\Nat$, $\Natp$, $\Int$ and $\Real$ denote the sets of integers,  positive integer, integer and real numbers, respectively.
For integers: $a|b$ if $a$ is a factor of $b$.
The $modulo$-function is used in the form 
$a \,mod\,b = a - b \cdot \lfloor \frac{a}{b} \rfloor$, which also holds for negative integers $a \in \Int$.
In particular, $-a\,mod\,b = b-a$ for $0<a\leq b$.


\section{Petri Space and Cycloids}  \label{sec-cycloids}

We define (Petri) nets as they will be used in this article. 

 \begin{definition}[\cite{Valk-2019}] \label{def-net} 
 As usual, a net $ \N{} = (S, T, F)$ is defined by non-empty, disjoint sets 
 $S$ of places and $T$ of transitions, connected by a flow relation 
 $F \subseteq (S \cp T) \cup (T \cp S)$ and $X := S \cup T$.
A transition $t \in T$ is \emph{active} or \emph{enabled} in a marking $M \subseteq S$ if $\; ^{\ndot} t \subseteq M \, \land \, t^{\ndot} \cap M = 
\emptyset$\footnote{With the condition $t^{\ndot} \cap M = \emptyset$ we follow Petri's definition, but with no impacts in this article.}.
In this case we obtain $M \stackrel{t}{\rightarrow}M'$ if $M' = M \backslash^{\ndot} t \cup t^{\ndot}$, where 
$^{\ndot} x := F^{-1}[x], \; x^{\ndot} := F[x] $
denotes the input and output elements of an element $x \in X$, respectively. 
 $\stackrel{*}{\rightarrow}$ is the reflexive and transitive closure of $\rightarrow$. 
\end{definition}

\begin{figure}
	\begin{center}
	\includegraphics [scale = 0.22]{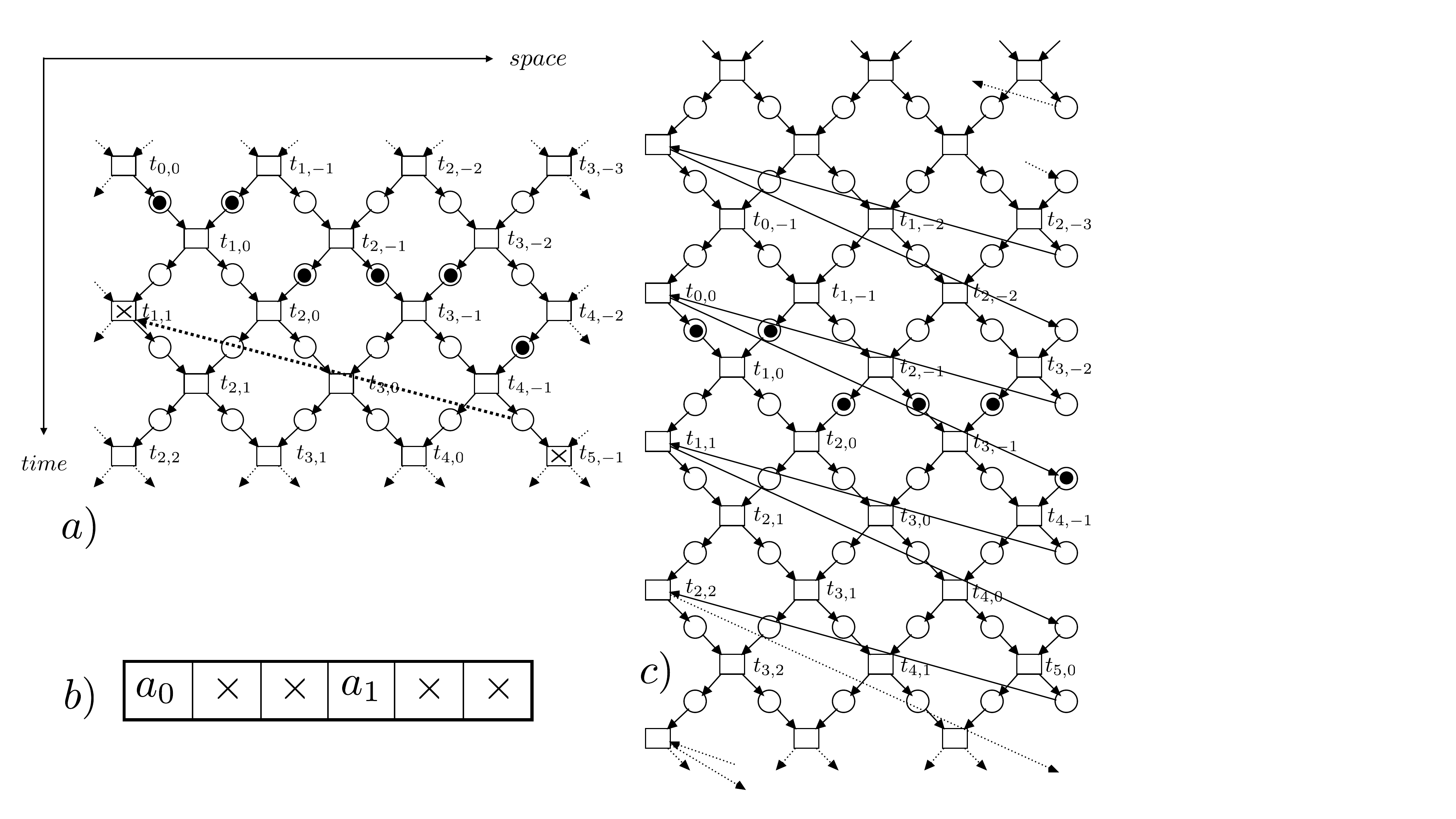}
		\caption{ a) Petri space, b) circular traffic queue and c) time orthoid.}
	\label{petrispace-key}
	 \end{center}
\end{figure}

Petri started with an event-oriented version of the Minkowski space which is called Petri space now. 
Contrary to the Minkowski space, the Petri space is independent of an embedding into $\Int \times \Int$.
It is therefore suitable for the modelling in transformed coordinates as in non-Euclidian space models.
However, the reader will wonder that we will apply linear algebra, for instance using equations of lines.
This is done only to determine the relative position of points.
It can be understood by first topologically transforming and embedding the space into $\Real \times \Real$, calculating the position and then transforming back into the Petri space.
Distances, however, are \underline{not} computed with respect to the Euclidean metric, but by counting steps in the grid of the Petri space, like Manhattan distance or taxicab geometry.

For instance, the transitions of the Petri space might model the moving of items in time and space in an unlimited way.
To be concrete, a coordination system is introduced with arbitrary origin (see Figure~\ref{petrispace-key} a).
The occurrence of transition $t_{1,0}$ in this figure, for instance, can be interpreted as a step of a traffic item (the token in the left input-place) in both space and time direction.
It is enabled by a gap or co-item (the token in the right input-place).
 Afterwads the traffic item can make a new step by the occurrence of transition $t_{2,0}$.
By the following definition the places obtain their names by their input transitions
(see Figure~\ref{P-space+FD} b).

\begin{figure}
	\begin{center}
	\includegraphics[width=0.8\textwidth]{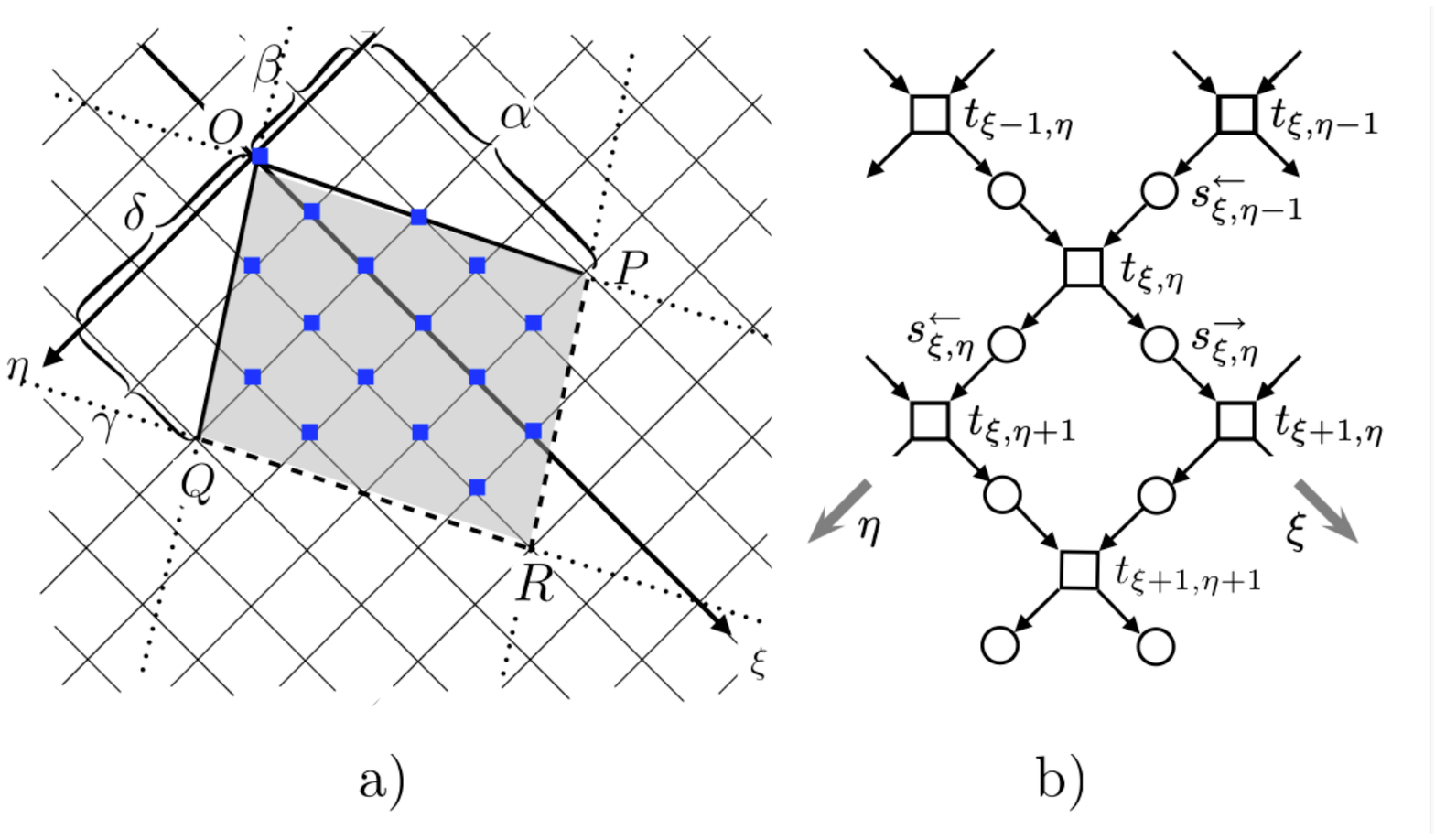}
		\caption{a)  Fundamental parallelogram of $\mathcal{C}(4,2,2,3)$ and  b) Petri space.}
		\label{P-space+FD}
	\end{center}
\end{figure}

\begin{definition} [\cite{Valk-2019}] \label{petrispace}
A $Petri \; space$ 
is defined by the net 
$\PS{1} := (S_1, T_1, F_1)$ \ where
$S_1 = \GSvw[1]\cup \GSrw[1], \;\GSvw[1]  =$ 
   $\col{\gsvw{\xi,\eta}}{\xi,\eta \in \Int},$ 
$\;\GSrw[1] = \col{\gsrw{\xi,\eta}}{\xi,\eta \in \Int }, \GSvw[1] \cap \GSrw[1] = \emptyset $, 
 $ T_1 =$ $ \col{t_{\xi,\eta}}{\xi,\eta \in \Int }, F_1 =$ $ \col{(t_{\xi,\eta},\gsvw{\xi,\eta})}{\xi,\eta \in \Int } \cup \col{(\gsvw{\xi,\eta},t_{\xi+1,\eta})}{\xi,\eta \in \Int } \cup $ 
 $\col{(t_{\xi,\eta},\gsrw{\xi,\eta})}{\xi,\eta \in \Int } \cup \col{(\gsrw{\xi,\eta},t_{\xi,\eta+1})}{\xi,\eta \in \Int }$ (cutout in Figure~\ref{P-space+FD} b).
 $\GSvw[1]$ is the set of \emph{forward places} and $\GSrw[1]$ the set of \emph{backward places}.
 $ \prenbfw{t_{\xi,\eta}}:= \gsvw{\xi-1,\eta}$ is the forward input place of $t_{\xi,\eta}$ and in the same way 
 $ \prenbbw{t_{\xi,\eta}}:= \gsrw{\xi,\eta-1}$,
 $ \postnbfw{t_{\xi,\eta}} := \gsvw{\xi,\eta}$ and
 $\postnbbw{t_{\xi,\eta}}:= \gsrw{\xi,\eta}$ (Figure~\ref{P-space+FD} b).

\end{definition}

In two steps, by a twofold folding with respect to time and space, Petri defined the cyclic structure of a cycloid.
 One of these steps is  
 a folding $f$ with respect to space with $f(i,k)=f(i+\alpha,k-\beta)$, fusing all points $(i,k)$ of the Petri space with $(i+\alpha,k-\beta)$ where $i,k \in \Int, \alpha,\beta \in \Natp$ (\cite{Petri-NTS}, page 37).
 While Petri gave a general motivation, oriented in physical spaces, we interpret the choice of $\alpha$ and $\beta$ by our model of traffic queues.

We assume that our model of a circular traffic queues has six slots containing two items $a_0$ and $a_1$ as shown in Figure~\ref{petrispace-key} b).
These are modelled in Figure~\ref{petrispace-key} a) by the tokens in the forward input places of $t_{1,0}$ and $t_{3,-1}$.
The four co-items (the empty slots in Figure \ref{petrispace-key} b) ) are represented by the tokens in the backward input places of $t_{1,0}, t_{2,0}$ and $t_{3,-1}, t_{4,-1}$.
By the occurrence of $t_{1,0}$ and $t_{2,0}$ the first item can make two steps, as well as the second item by the transitions $t_{3,-1}$ and $t_{4,-1}$, respectively.
Then $a_1$ has reached the end of the queue and has to wait until the first item is leaving its position.
Hence, we have to introduce a precedence restriction between the transitions $t_{1,0}$ and $t_{5,-1}$.
This is done by fusing the transitions 
 $t_{5,-1}$ and the left-hand follower $t_{1,1}$ of $t_{1,0}$ ,  which are marked by a cross in Figure~\ref{petrispace-key} a).  This is implemented by the dotted arc in the same figure.
To determinate $\alpha$ and $\beta$ we set 
 $(5,-1) =(1+\alpha,1-\beta)$ which gives 
 $5 = 1+ \alpha$ or $\alpha = 4$ and 
 $-1 = 1- \beta$ or $\beta = 2$.
By the equivalence relation
 $t_{\xi,\eta} \equiv t_{\xi + 4,\eta -2} $ we obtain the structure in Figure~\ref{petrispace-key} c). 
The resulting still infinite net is called a \emph{time orthoid} (\cite{Petri-NTS}, page 37), as it extends infinitely in temporal future and past.
The second step is a folding with $f(i,k)=f(i+\gamma,k+\delta)$ with $\gamma,\delta \in \Natp$ reducing the system to a cyclic structure also in time direction. 
As shown in \cite{Valk-2020} an equivalent cycloid for the traffic queue of Figure~\ref{petrispace-key} b) has the parameters $ ( \alpha, \beta, \gamma, \delta ) = (4,2,2,2)$.
To keep the example more general, 
in Figure~\ref{P-space+FD} a) the values 
 $ ( \alpha, \beta, \gamma, \delta ) = (4,2,2,3)$ are chosen.
 In this representation of a cycloid, called \emph{fundamental parallelogram}, the squares of the transitions as well as the circles of the places are omitted.
 All transitions with coordinates within the parallelogram belong to the cycloid including those on the lines between $O,Q$ and $O, P$, but excluding those of the points $Q,R,P$ and those on the dotted edges between them.
 All parallelograms of the same shape, as indicated by dotted lines outside the fundamental parallelogram are fused with it.

\begin{definition} [\cite{Valk-2019}] 
\label{cycloid}
A \emph{cycloid} is a net \ $ \zyk( \alpha, \beta, \gamma, \delta ) = (S, T, F)$, defined by parameters 
 \ $ \alpha, \beta, \gamma, \delta \in \Natp$, by a quotient  of the Petri space \ $\PS{1} := (S_1, T_1, F_1)$ \ 
with respect to the equivalence relation
$\mo\zykaeq 
\subseteq X_1 \cp X_1 $
with $X_1 = S_1 \cup T_1$,
 $ \mo\zykaeq[\GSvw[1]] \subseteq \GSvw[1], 
 \mo\zykaeq[\GSrw[1]] \subseteq \GSrw[1],
 \mo\zykaeq[T_1] \subseteq T_1,$ 
 $ x_{\xi,\eta} \zykaeq x_{\xi+m\alpha+n\gamma,\,\eta-m\beta+n\delta} $
for all $ \xi, \eta, m, n \in \Int $\,, $ X = X_1/_\zykaeq $,
 $ \eqcl[\zykaeq]{x} \mb{F} \eqcl[\zykaeq]{y} \: \Leftrightarrow
\exists\,{x'\in\eqcl[\zykaeq]{x}}\,\exists\,y' \in \eqcl[\zykaeq]{y}: \,x' F_1 y' $
 \ for all $x, y \in X_1 $. 
The matrix $\mathbf{A} = \begin{pmatrix} \alpha & \gamma \\ -\beta & \delta \end{pmatrix} $ is called the matrix of the cycloid.
Petri denoted the number $|T|$ of transitions as the area $A$ of the cycloid and proved in \cite{Petri-NTS} its value to $|T| =A =\alpha\delta+\beta\gamma$ which equals the determinant $A = det(\mathbf{A})$.
The embedding of a cycloid in the Petri space is called \emph{fundamental parallelogram} 
(see Figure~\ref{P-space+FD} a).
\end{definition}

\begin{theorem}\label{symmetry}
The following cycloids are isomorphic\footnote{By a net isomorphism \cite{Reisig-Smith-87}.} to $\zyk(\alpha,\beta,\gamma,\delta) $:\\
    $\;\;\;$ a) $\zyk(\beta,\alpha,\delta,\gamma) $, $\;\;\;\;\;\;\;\;\;\;\;($The \emph{dual cycloid} of $\mathcal{C}(\alpha,\beta,\gamma,\delta).)$\\
     $\;\;\;$ b) $\zyk(\alpha,\beta,\gamma - q \cdot \alpha,\delta+q \cdot \beta) $ if $ q \in \Natp $ and  $\gamma > q \cdot  \alpha$, \\
      $\;\;\;$ c) $\zyk(\alpha,\beta,\gamma + q \cdot \alpha,\delta- q \cdot\beta) $ if $  q \in \Natp $ and $\delta > q \cdot \beta$.
\end{theorem}

For proving the equivalence of two points in the Petri space the following procedure\footnote{The algorithm is implemented under \url{http://cycloids.de}.} is useful.

\begin{theorem}[\cite{Valk-2020}\cite{valk-arXiv-algebra-2024}] \label{parameter}
Two points $\vec{x}_1, \vec{x}_2\in X_1$ are equivalent $\vec{x}_1 \equiv \vec{x}_2$ if and only if
 for the difference $\vec{v} := \vec{x_2}-\vec{x_1}$
 the parameter vector 
 $\pi(\vec{v}) = \frac{1}{A} \cdot\mathbf{B} \cdot \vec{v}$ has integer values, where $A$ is the area and 
 $\mathbf{B} = \begin{pmatrix} \delta & -\gamma \\ \beta & \alpha \end{pmatrix}$. \\
 In analogy to Definition \ref{cycloid} we obtain 
 $\vec{x}_1 \equiv \vec{x}_2 \Leftrightarrow$ 
 $\exists \;m, n \in \Int: \vec{x_2}-\vec{x_1} =\mathbf{A}\begin{pmatrix} m \\ n \end{pmatrix}$.
\end{theorem}

\begin{lemma} [\cite{Valk-2019}] \label{normal form}
For any cycloid $\zyk(\alpha,\beta,\gamma,\delta) $ there is a minimal cycle containing the origin $O$ in its fundamental parallelogram representation.
\end{lemma}

\begin{theorem} \label{th-f-b-cycle}
In a cycloid   $ \mathcal{C}( \alpha, \beta, \gamma, \delta ) $ with area $A$ the length of a forward-cycle is 
$p = \frac{A}{gcd(\beta,\delta)} $ and length of a backward-cycle is $p' =\frac{A}{gcd(\alpha,\gamma)}$. 
\end{theorem} 

\begin{proof} By the symmetry of a cycloid all  forward-cycles have the same length.  Therefore it is sufficient to consider the  forward-cycle starting in the origin $(0,0)$. It is given by proceeding on the $\xi-$axis until  for the first time
a point $(\xi,0)$ is obtained, which is equivalent to the origin. By Theorem \ref{parameter} a necessary and sufficient condition for 
$\begin{pmatrix}0 \\ 0  \end{pmatrix}  \equiv \begin{pmatrix} \xi \\  0 \end{pmatrix}$ is 
$\pi(\begin{pmatrix}\xi \\ 0  \end{pmatrix}  -  \begin{pmatrix} 0\\  0 \end{pmatrix}) = 
\frac{1}{A} \begin{pmatrix} \delta & -\gamma \\ \beta & \alpha \end{pmatrix}  \begin{pmatrix}\xi \\ 0  \end{pmatrix}  =
\frac{1}{A}  \begin{pmatrix}\delta \cdot \xi \\ \beta \cdot \xi   \end{pmatrix} \in \Int^2$. This is equivalent to 
$\frac{\delta}{A} \cdot \xi \in \Int \; \land \frac{\beta}{A} \cdot \xi \in \Int $. $\xi = A$ is a solution in $\frac{\delta}{A} \cdot \xi \in \Int $, but $\xi = \frac{A}{gcd(A,\delta)} $ is minimal. The same with  $\xi = \frac{A}{gcd(A,\beta)} $ for $\frac{\beta}{A} \cdot \xi \in \Int $ and both together give
$ \xi = \frac{A}{\omega} $ with $\omega = gcd(gcd(\beta,A),gcd(\delta,A))$.  To finish the proof it is sufficient to prove that $\omega$ equals
$ gcd(\beta,\delta) $. To this end we first prove: 
\begin{equation}\label{Gl1}
gcd(\beta,\alpha \cdot \delta + \beta \cdot \gamma) =   gcd(\beta,\alpha \cdot \delta)
 \end{equation}
 In fact, if $q$ divides $\beta$ and $\alpha \cdot \delta + \beta \cdot \gamma$ then $q$ divides $\alpha \cdot \delta$.
Conversely, if $q$ divides $\beta$ and $\alpha \cdot \delta$ then $q$ divides $\alpha \cdot \delta + \beta \cdot \gamma$.
In the same way the following equation holds.
\begin{equation}\label{Gl2}
gcd(\delta,\alpha \cdot \delta + \beta \cdot \gamma) =   gcd(\delta,\beta \cdot \gamma)
 \end{equation}
 The next equation
 \begin{equation}\label{Gl3} 
 gcd(gcd(\beta,\alpha \cdot \delta),gcd(\delta,\beta \cdot \gamma)) = gcd(\beta,\delta)
 \end{equation}
  is proved by the following true logical formula: 
  $q | \beta  \land q | \alpha \cdot \delta \land q | \delta  \land q | \beta \cdot \gamma \; \Leftrightarrow \; q | \beta  \land  q | \delta $.
  Using these equations we obtain: \\
  $\omega 
  {\underset{\text{}}{=}} 
  gcd(gcd(\beta,\alpha \cdot \delta + \beta \cdot \gamma),gcd(\delta,\beta,\alpha \cdot \delta + \beta \cdot \gamma))  {\underset{\text{(\ref{Gl1})(\ref{Gl2})}}{=}} \\
  gcd(gcd(\beta,\alpha \cdot \delta),gcd(\delta,\beta \cdot \gamma))  {\underset{\text{(\ref{Gl3})}}{=}} gcd(\beta,\delta)$.
  The case of backward-cycles is proved in a similar way. Alternatively, we use the isomorphism of Theorem \ref{symmetry} a)
   $ \mathcal{C}( \alpha, \beta, \gamma, \delta )$ \\ $ \simeq  \mathcal{C}(  \beta, \alpha, \delta, \gamma ) $ and derive the second part from the first part by the substitution $\alpha \mapsto \beta, \beta \mapsto \alpha, \gamma \mapsto \delta,\delta \mapsto \gamma $.
 \qed\end{proof}

\begin{definition} [\cite{Valk-2019}]  \label{standard-M0}
For a cycloid $\zyk(\alpha,\beta,\gamma,\delta)$ 
we define a cycloid-system\\ 
$\mathcal{C}(\alpha,\beta,\gamma,\delta,M_0)$ 
or $\mathcal{C}(\N{},M_0)$ by adding the standard initial marking: 

$ M_0 =$  $ \col{ \gsvw{\xi,\eta} \in \GSvw[1] }{\, \beta\xi + \alpha\eta \,\leq\, 0 \  \land \  \beta(\xi+1) + \alpha\eta \,>\, 0  }/_\zykaeq \ \,\cup $\\
$\; \; \; \;\; \; \; \; \; \; \; \; \; \;\; \;\col{\gsrw{\xi,\eta} \in \GSrw[1] }{\, \beta\xi + \alpha\eta \,\leq\, 0 \   \land \  \beta\xi + \alpha(\eta+1) \,>\, 0 }/_\zykaeq $. \\
 Note, that by \cite{Valk-2019}  a standard initial marking contains $\beta$ tokens in forward and $\alpha$ tokens in backward places.
\end{definition}
 
 The motivation of this definition is given in  \cite{Valk-2019}. See Figure \ref{c-4-3-3-3}b) for an example of a cycloid with standard initial marking. 

 \begin{definition} \label{regular-M0}
 For a cycloid $\mathcal{C}(\alpha,\beta,\gamma,\delta)$ a regular initial marking is defined by a number of $\beta$ forward places 
 $\{\gsvw{-1,i} | \; 0 \geq i > -\beta\}$ and a number of $\alpha$ backward places 
 $\{\gsrw{i,-\beta} | \; 0 \leq i < \alpha\}$.
 Note, that a regular initial marking contains $\beta$ tokens in forward and $\alpha$ tokens in backward places.
 \end{definition}
 
 \begin{corollary} \label{number_of_cycles}
A   cycloid  $ \mathcal{C}( \alpha, \beta, \gamma, \delta) $  contains $gcd(\beta,\delta)$ disjoint forward cycles and  $gcd(\alpha,\gamma)$ disjoint backward cycles.
In a cycloid system $ \mathcal{C}( \alpha, \beta, \gamma, \delta, M_0 ) $  with standard or regular initial marking $M_0$ the number of tokens in a forward cycle is $\frac{\beta}{gcd(\beta,\delta)} $ and  $\frac{\alpha}{gcd(\alpha,\gamma)} $  in an backward cycle.
\end{corollary} 

\begin{proof} 
Since all forward cycles are disjoint with a number of $\frac{A}{gcd(\beta,\delta)} $ transitions (Theorem \ref{th-f-b-cycle}), the number of  forward cycles is
$\frac{A}{(\frac{A}{gcd(\beta,\delta)} )} = gcd(\beta,\delta) $. Similarly, for backward cycles the number is 
$\frac{A}{(\frac{A}{gcd(\alpha,\gamma)}) } = gcd(\alpha,\gamma) $. Since in a cycloid system there are $\beta$ tokens in forward places (Definitions \ref{standard-M0}, \ref{regular-M0}), each forward cycle contains $\frac{\beta}{gcd(\beta,\delta)} $ tokens. In the same way, since there are $\alpha$ tokens in  backward places the number of tokens in a backward cycle is $\frac{\alpha}{gcd(\alpha,\gamma)} $.
\qed\end{proof}

 When working with cycloids it is sometimes important to find for a transition outside the fundamental parallelogram the equivalent element inside.  
For instance the first set in Definition \ref{regular-M0} of a regular initial marking contains the element $\gsvw{-1,1 - \beta}$. It is named by its input transition $t_{-1,1 - \beta}$, which is outside the fundamental parallelogram. To obtain the corresponding place of the cycloid net the equivalent transition of $t_{-1,1 - \beta}$ inside the fundamental parallelogram has to be computed.
In general,
by enumerating all elements of the fundamental parallelogram  and applying the equivalence test from Theorem \ref{parameter}  a runtime is obtained, which already fails for small cycloids. The following theorem allows for a better algorithm, which is linear with respect to the cycloid parameters.

\begin{theorem} [\cite{Valk-2020}]\label{to_FP}
For any element $\vec{u} = (u,v)$ of the Petri space the (unique) equivalent element  of the fundamental parallelogram is 
$\vec{x} = \vec{u} - \mathbf{A}\begin{pmatrix} m \\ n  \end{pmatrix}$ where \\
$m = \lfloor  \frac{1}{A}(u\delta - v \gamma)\rfloor$ and $n = \lfloor  \frac{1}{A}(v\alpha + u\beta)\rfloor$.
\end{theorem}

\vspace{-0.2  cm}
\section{Regular Cycloids}    \label{sec-regular-cyc}

Circular traffic queues are composed  by a number  $c$ of sequential and interacting processes of equal length. In the formalism of cycloids this corresponds to a number of $\beta$ disjoint forward cycles  of equal length $p$. Cycloids with such a property are called \emph{regular}.

\begin{definition} \label{regular-cycloid}
 A cycloid system $ \mathcal{C}=\mathcal{C}(\alpha,\beta,\gamma,\delta, M_0)$ is called \emph{regular} if  each forward cycle contains exactly one token with respect to its (standard or regular) initial marking. The forward cycle is called \emph{process}
 in this case (also sometimes item-process, car-process). 
  $ \mathcal{C}$ is called \emph{co-regular} if  each backward cycle contains exactly one token with respect to its (standard or regular) initial marking.
 The backward cycle is called \emph{co-process}
 in this case (also sometimes co-item-process, co-car-process).
\end{definition}

\begin{theorem}   \label{th-regular}
A cycloid system $\mathcal{C}=\mathcal{C}(\alpha,\beta,\gamma,\delta,M_0)$  is regular if and only if $ \beta | \delta $. In this case it contains a number of $\beta$ processes, each  of length $p = \frac{A}{\beta} $.
$\mathcal{C}$  is co-regular if and only if $ \alpha | \gamma$. In this case it contains a number of $\alpha$ co-processes, each  of length $p' = \frac{A}{\alpha} $. 
\end{theorem}
\begin{proof}
By Corollary  \ref{number_of_cycles},
if a forward cycle contains exactly one token  we obtain
$\frac{\beta}{gcd(\beta,\delta)} = 1 \; \Leftrightarrow \;  \beta =
 gcd(\beta,\delta) \; \Leftrightarrow \; \beta | \delta $. By Theorem  \ref{th-f-b-cycle} the length of processes is 
 $p = \frac{A}{gcd(\beta,\delta)}  = \frac{A}{\beta}$. 
 Similarly, for a co-process $\frac{\alpha}{gcd(\alpha,\delta)} = 1 \; \Leftrightarrow \;  
\beta | \delta $. 
The length of a co-process is $p' =\frac{A}{gcd(\alpha,\gamma)} =\frac{A}{\alpha}$.
By  \cite{Valk-2019} a standard or  initial marking contains a number of $\beta$ tokens in forward places and 
$\alpha$ tokens in backward places. The same follows immediately for a regular initial marking  from Definition \ref{regular-M0}.
\qed\end{proof}

%
\begin{figure}[t]
  \includegraphics[width=1\textwidth]{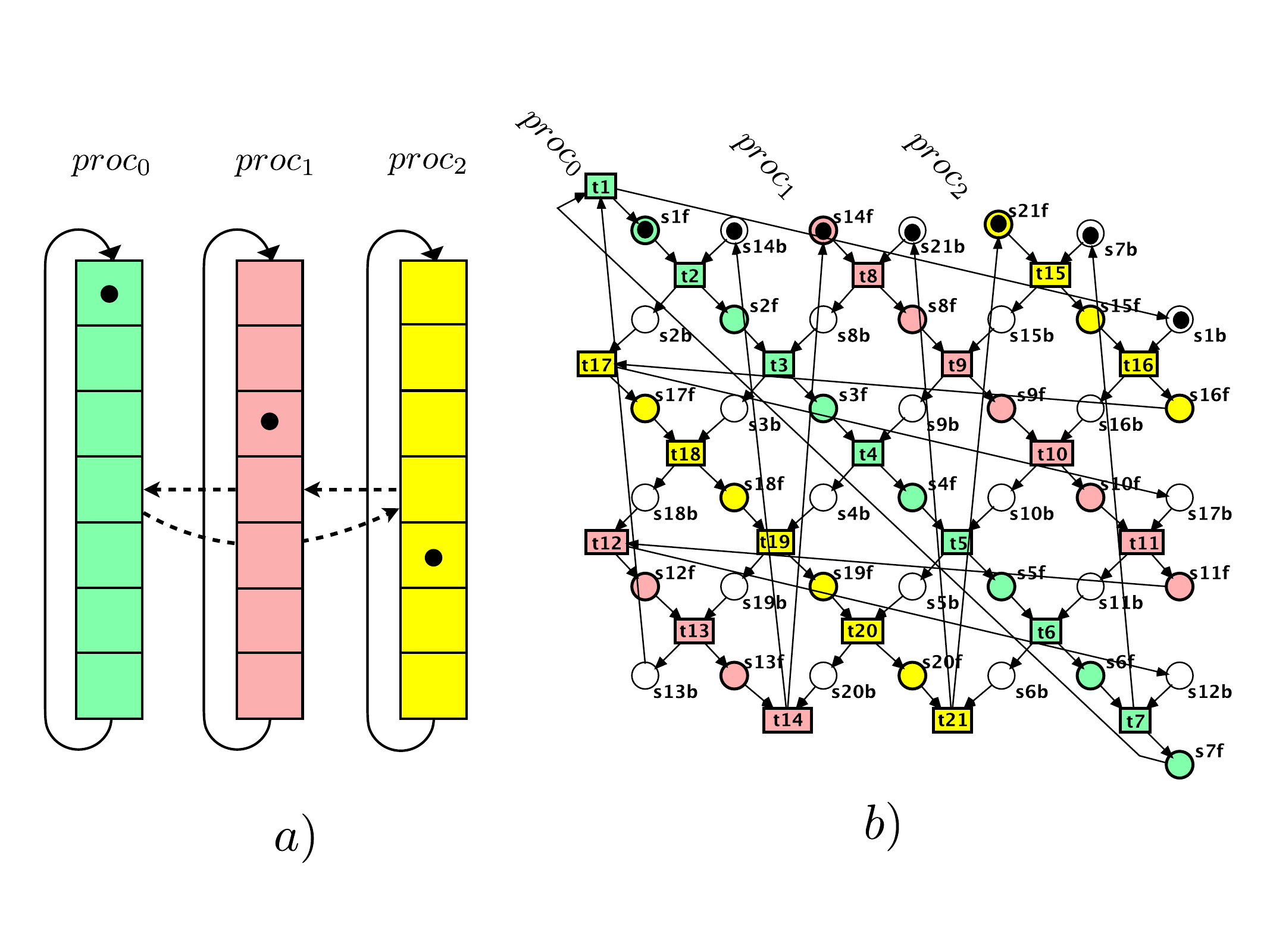}
                     \caption{Cycloid $\mathcal{C}(4,3,3,3)$ in a) and with regular coordinates in b).}
                             \label{c-4-3-3-3}
                                  \end{figure}

To exploit the structure of a regular cycloid we
 define specific coordinates, called \emph{ regular coordinates}. 
 The process of  a traffic item $a_0$ starts with transition $t_{0,0}$ which is denoted $[t_0,a_0]$, having the input place $[s_{p-1},a_0]$. The next transitions are $[t_1,a_0]$ up to $[t_{p-1},a_0]$ and then returning to $[t_0,a_0]$. The other processes for $a_1$ to $a_{c-1}$ (with $\beta = c$) are denoted in the same way (see Figure \ref{c-4-3-3-3} b). 
 As the process of $a_j$ starts in position $j$ of the queue, its initial token is in $[s_{j\ominus_p  1},a_j]$.

\begin{definition}  \label{d-reg-coord}
Given a regular cycloid  $\mathcal{C}(\alpha,\beta,\gamma,\delta), $\emph{ regular coordinates} are defined as follows: transitions of a $a_j$-process $0  \leq j < \beta$, each with length $p$, are denoted by $\{[t_0,a_j],\cdots,[t_{p-1},a_j]\}$. For each transition $[t_{i},a_{j}]$ we define   the output places by
$ \postnbfw{[t_{i},a_{j}]} := [s_i,a_j]$
and  $ \postnbbw{[t_{i},a_{j}]}:= [s'_i,a_j]$  and the output transition by
 $[s_i,a_j]{}^{\ndot} := [t_{i\oplus_p 1},a_j]$ for $0  \leq i < p, 0  \leq j < c$.
 Regular coordinates are related to standard coordinates of the Petri space by defining the following initial condition
$stand([t_0,a_j]) := t_{-j,-j}$ for $0 \leq j < c$ (taking the equivalent transition of $ t_{-j,-j}$ in the fundamental parallelogram). 
\end{definition}

For instance, in  Figure \ref{c-4-3-3-3}  b) we obtain for the last formula in Definition \ref{d-reg-coord}: $[t_{0},a_{2}] := t_{-2,-2} \equiv t_{1,1}$. While the output place $[s'_i,a_j]$ in regular coordinates takes its name from the input transition in Definition \ref{d-reg-coord}, it remains to determine its output transition according to the corresponding regular coordinates. 


\begin{lemma} \label{le-reg-coord}
In a regular cycloid the injective mapping  \emph{stand} from regular to standard coordinates is given by $stand([t_i,a_j])= t_{i-j,-j}$ 
for  $0 \leq i < p$ and $0 \leq j < \beta$ 
(modulo equivalent transitions). 
The output transition of  $ [s'_i,a_j]$ is \\
a) $ [s'_i,a_0]{}^{\ndot} =[t_{(i+\beta+\alpha-1)\,mod\,p},a_{\beta-1}]$  for $j = 0$ and \\
b)  $[s'_i,a_j]^{\ndot} = [t_{i \ominus_p 1},a_{j \ominus_\beta  1}]$ for $0< j < c$. \\
c) If $p = \alpha + \beta$ the two cases coincide.
\end{lemma}

\begin{proof}
For a given $j$ by Definition \ref{d-reg-coord} we have $stand[t_0,a_j] := t_{-j,-j}$. Adding a value 
$i \in  \{ 0,\cdots,p-1 \} $ to the index of $t_0$ we obtain the index of 
$t_i$, hence $stand([t_i,a_j]) := t_{-j+i,-j}$.\\
a)
By the preceding result $stand([t_i,a_0]) := t_{i,0}$.
To prove a) we observe that in the fundamental parallelogram from $[t_i,a_0]) = {}^{\ndot}{[s'_i,a_0]}$ we should come to $[t_{(i+\beta+\alpha-1)\,mod\,p},a_{\beta-1}]$ by taking one step in the $\eta$-direction.
Therefore it is sufficient to prove $stand([t_i,a_0])  + (0,1) \equiv stand([t_{(i+\beta+\alpha-1)},a_{\beta-1}])$ or \\
$\begin{pmatrix} i \\ 0  \end{pmatrix} + \begin{pmatrix} 0\\ 1  \end{pmatrix} \equiv $ 
$\begin{pmatrix} i+\beta+\alpha-1-(\beta+1)\\ -(\beta-1)  \end{pmatrix}$ or 
$\begin{pmatrix} i \\ 1 \end{pmatrix}  \equiv \begin{pmatrix} i+\alpha\\ 1-\beta \end{pmatrix}$. \\
$\pi(\begin{pmatrix}i\\ 1  \end{pmatrix}-\begin{pmatrix}i+\alpha\\ 1-\beta  \end{pmatrix}) = 
\pi(\begin{pmatrix}-\alpha\\ \beta  \end{pmatrix}) = 
\frac{1}{A} \begin{pmatrix} \delta & -\gamma \\ \beta & \alpha \end{pmatrix}\begin{pmatrix}-\alpha \\ \beta  \end{pmatrix} = 
 \frac{1}{A} \begin{pmatrix} -A\\ 0 \end{pmatrix}  
  \in \Int \times \Int $. \\
  b) The same method results in the following equivalence to be proved:\\
  $stand([t_i,a_j])  + (0,1) \equiv stand([t_{(i-1)},a_{j-1}])$ or
$\begin{pmatrix} i-j \\ -j  \end{pmatrix} + \begin{pmatrix} 0\\ 1  \end{pmatrix} \equiv
\begin{pmatrix} i-1-(j-1)\\ 1-j  \end{pmatrix} $ which is obvious (without using Theorem \ref{parameter}).\\
c) If $p=\alpha+\beta$ then $(i+\beta+\alpha-1)\,mod\,p = (i-1)\,mod\,p$
\qed
\end{proof}


\begin{corollary} \label{regular-M0-reg-coord}
a) The regular initial marking of a regular cycloid system  \\
$\mathcal{C}(\alpha,\beta,\gamma,\delta,M_0)$ with  process length $p$ in regular coordinates is \\
$M_0 = \{[s_{p-1},a_0]\} \cup \{ [s_i,a_{i+1}] | 0\leq i < \beta-1  \}   \cup 
 \{ [s'_i,a_0] | p-\alpha\leq i < p\}$.\\
 For later reference, we note $ \prenbbw{[t_{\beta-1},a_{\beta-1}]} 
= [s'_{p-\alpha},a_0]$.\\
b) From the regular initial marking $M_0$ for each $k\in\{ 0,\cdots,p-1 \}$ the marking 
$M_k = \{[(s_{p-1+k)\; mod  \; p},a_0]\} \cup \{ [s_{i\oplus_p k},a_{i+1}] | 0\leq i < \beta-1  \}   \cup \\
 \{ [s'_{i\oplus_p k},a_0] | p-\alpha\leq i < p\}$ is reachable by $k\cdot\beta$ transition occurrences. 
 $M_k$ is called a $k$-\emph{regular} or simply \emph{regular} marking of the cycloid.
 \end{corollary}

\begin{proof}
a) Since $p-1 = (-1) \; mod \; p$ we can write $ \{ [s_i,a_{i+1}] | -1\leq i < \beta-1  \} $ for the forward places of $M_0$.
As the mapping $stand$ is  defined on transitions, we go to the input transitions and apply $stand$ to obtain
$ \{ [t_i,a_{i+1}] | -1\leq i < \beta-1  \} $ and
$  \{ stand([t_i,a_{i+1}] ) | -1\leq i < \beta-1  \} 
=  \{ t_{-1,-(i+1)} | -1\leq i < \beta-1  \} 
=  \{ t_{-1,i} |\; 0 \geq i > -\beta\} \}$, 
which is the same as the set of input transitions of the forward places in 
Definition \ref{regular-M0}.\\
Since the mapping $stand$ is injective we can conclude also in the inverse direction.
%
%
To prove the second part of the union recall that the last  traffic item $a_{\beta-1}$ in a regular initial marking is active (enabled). Therefore also the backward input place  $[s'_i,a_0]$ of $[t_{\beta-1},a_{\beta-1}]$ must be marked. Using Lemma \ref{le-reg-coord}  the value of $i$ must satisfy $ [s'_i,a_0]{}^{\ndot} =$ 
$[t_{(i+\alpha+\beta-1)\,mod\,p},a_{\beta-1}] = [t_\beta,a_\beta]$, and for $i$ we obtain the following condition: $(i+\alpha+\beta-1)\,mod\,p = \beta-1$ and $i = ( - \alpha) \, mod \, p = p - \alpha $. 
This holds since $\beta|\delta \; \Rightarrow \;\beta \leq \delta$ and therefore $p = \frac{A}{\beta} = \alpha\frac{\delta}{\beta}+\gamma > \alpha$.
The marked place in question is therefore
$[s'_i,a_0] = [s'_{ p - \alpha },a_0]$. 
To determine the other elements of $ \{ [s'_i,a_0] | p-\alpha\leq i \leq p  \}$ recall that the last  traffic item $a_{\beta-1}$ should be
 able to make $\alpha$ steps before any other transition has to occur. Therefore also the places $[s'_{ p - \alpha + 1},a_0]  $ to $[s'_{ p - \alpha + \alpha-1},a_0]  $
 must be marked in the regular initial marking.\\
 b) In the marking $M_k$ transitions $[t_i,a_{\beta-1}],[t_{i\ominus_p 1 },a_{\beta-2}],\cdots,[t_{i\ominus_p \beta },a_{0}]$ for some $0 \leq i < p$ can occur and the indices of the forward markings are increased by $1$.
The token in the first place $ [s'_{(p-\alpha + k) \; mod  \; p},a_0] $ in the backward places of $M_k$ is removed by $[t_i,a_{\beta-1}]$. The last transition $[t_{i\ominus_p \beta },a_{0}]$ adds a token to 
$[s_{(p-1+k+1)\; mod  \; p},a_0]$ and therefore also to $[s'_{(p-1+k+1)\; mod  \; p},a_0]$.
Therefore the set $\{ [s'_{i\oplus_p k},a_0] | p-\alpha\leq i < p\}$ of $M_k$ is transformed to the corresponding set $\{ [s'_{(i + k +1)\; mod  \; p},a_0] | p-\alpha\leq i < p\}$ of $M_{k+1}$.
 \qed
\end{proof}

 In the regular cycloid system $\mathcal{C}(4,3,3,3,M_0)$ in Figure \ref{c-4-3-3-3} b) we obtain $ \prenbbw{[t_{2},a_{2}]}\\
= [s'_{p-\alpha},a_0]  = [s'_{7-4},a_0]$. 
The given regular initial  marking is 
 $\{[s_6,a_0],[s_0,a_1],$ $[s_1,a_2],$ $[s'_3,a_0],[s'_4,a_0],
 [s'_5,a_0],[s'_6,a_0]\}$.  The standard initial marking is  given by bold bordered circles.

\begin{lemma} \label{bw-input-place}
In a regular cycloid for $0 \leq j < p$ :
\begin{itemize}
       \item [a)] $\prenbbw{[t_{i},a_{\beta - 1}]} = [s'_{i \ominus_p  (n-1)},a_0]$ and 
                      $\prenbbw{[t_{i},a_{\beta - 1}]} = [s'_{i \oplus_n 1)},a_0]$ for $p = n$.
       \item [b)] $\prenbbw{[t_{i},a_{j}]} = [s'_{i \oplus_p 1},a_{j+1}]$ for $0  \leq j < \beta -1$.
\end{itemize}     
\end{lemma} 
\begin{proof} 
a) By Lemma \ref{le-reg-coord} a)$ [s'_k,a_0]{}^{\ndot} = [t_{(k+n-1)\; mod  \; p},a_{\beta-1}]$, hence
$\prenbbw{[t_{i},a_{\beta - 1}]} = [s'_k,a_0] $ with $i = (k+n-1)\; mod  \; p$ and $k =(i-n+1)\; mod \; p = i \ominus_p  (n-1) $. 
If $p = n$ then $k =(i-n+1)\; mod \; n = ((i+1)-n)\; mod \; n = i \oplus_n 1 $.\\
b) By Lemma \ref{le-reg-coord} b) $ [s'_k,a_j]{}^{\ndot} = [t_{k \ominus_p  1},a_{j\ominus_\beta  1}]$ for $0 < j < \beta$. It follows 
$\prenbbw{[t_{i},a_{j}]} = [s'_{k},a_{j+1}]$ with $i = k \ominus_p  1$, hence $k = i \oplus_p 1$. 
\qed\end{proof}

\begin{lemma} \label{le-reg-cyc}
Let be $\mathcal{C}(\alpha,\beta,\gamma,\delta)$ a regular cycloid with process length $p$ and  minimal cycle length $cyc$. Then
  \begin{itemize}
       \item [a)] $cyc = p$  if $\alpha \leq \beta$,
       \item [b)]   $cyc = \frac{\beta}{\alpha}\cdot p$ if $\alpha > \beta$ and $\alpha | p$,
       \item [c)] $cyc = 2 \cdot \beta $  if $\alpha > \beta = \gamma = \delta$. 
  \end{itemize}     
\end{lemma} 
\begin{proof}
a) By Theorem 7 from \cite{valk-arXiv-algebra-2024} and $\delta=m\cdot \beta$ for some $m \in \Int$ we obtain \\
$cyc = \gamma + \delta +  \lfloor  \frac{\delta}{\beta}\rfloor (\alpha - \beta) =
\gamma + m\cdot \beta +  \lfloor  \frac{m\cdot \beta}{\beta}\rfloor  (\alpha - \beta)  =  
\gamma + m\cdot \beta +  m  (\alpha - \beta)  =  
\gamma + m\cdot \alpha$. 
This term equals 
$p = 
\frac{A}{\beta} = 
\frac{1}{\beta}  (\alpha\delta+\beta\gamma) = 
\frac{1}{\beta}  (\alpha\cdot m\cdot \beta+\beta\gamma)  =   \alpha \cdot m + \gamma$.

b) If $\alpha > \beta$ and 
$\alpha | p  =
\frac{A}{\beta} = 
\frac{1}{\beta}  (\alpha\delta+\beta\gamma) =
\alpha \cdot \frac{\delta}{\beta} + \gamma 
$
then $\alpha | \gamma$ and
case a) applies to the dual cycloid $\mathcal{C}(\beta,\alpha,\delta, \gamma)$ (Definition \ref{symmetry}) which is regular since $\alpha | \gamma$. Hence $cyc = p'$ where $p' = \frac{A}{\alpha} = \frac{\beta\cdot p}{\alpha}$ is the process length of the dual cycloid. Since it  is isomorphic to $\mathcal{C}(\alpha,\beta,\gamma,\delta)$ by Theorem \ref{symmetry} it has the same value $cyc = p' = \frac{\beta}{\alpha}\cdot p$.

c) If $\alpha > \beta = \gamma = \delta$ then $cyc = \gamma + \delta - \lfloor  \frac{\gamma}{\alpha}\rfloor  (\alpha - \beta) = \beta + \beta - 0 \cdot (\alpha - \beta)$. If $\alpha = \beta = \gamma = \delta$ then $cyc = \gamma + \delta + \lfloor  \frac{\delta}{\beta}\rfloor  (\alpha - \beta) = \beta + \beta + 1 \cdot 0$.
\qed
\end{proof}

The regular cycloid $\mathcal{C}(4,3,3,6)$ does not satisfy any of the conditions of Lemma \ref{le-reg-cyc}. The parameters in question are $cyc = 9$, 
$p= 11$ and $ \frac{\beta}{\alpha}\cdot p = \frac{3}{4}\cdot 11$.


\vspace{-0.4  cm}
\section{Stop-resilient Cycloids}\label{sec-resilient}

Considered as cooperating processes cycloids perform a  strong synchronization regimen. Therefore it is surprising that by a small extension we can model these processes to be stoppable or failing without stopping the other processes. 
As a byproduct it is proved that by eliminating the car with the highest index $c-1 =\beta-1$ the cycloid for one car less and one gap more is obtained. This corresponds to a step from a circular traffic queue with gaps $tq_1(c,g)$
to the model $tq_1(c-1,g+1)$.
Since the extension is defined by a folding of the backward places  only, the forward places and transitions of the processes are not modified and the cycloid algebra 
\cite{valk-arXiv-algebra-2024} can be applied.
To obtain a live system when one car is stopped we require that at least $\beta>1 $ cars are present.


\begin{definition} \label{de-bf-folding} For a given regular cycloid system $ \mathcal{C} =\mathcal{C}( \alpha, \beta, \gamma, \delta, M_0 ) $ with $\beta >1$, process length $p$ and a fixed set $D \subseteq\{0,\cdots,\beta-1\} $ with $|D| >1$,
called the set of \emph{back indices},
 we define the \emph{ backward folding}  
$ \mathcal{C}_{bf(D)}( \alpha, \beta, \gamma, \delta, \eqcl[D]{M_0} ) $ 
by a relation $\equiv_{bf(D)}$  on the backward places 
$\postnbfw{[t_{i},a_{j}]} = [s'_{i},a_{j}]$ by
\begin{equation}\label{def-bf}
[s'_{i},a_{j}]  \;\;\; \equiv_{bf(D)}   \;\; [s'_{r},a_{s}] \; \;\Leftrightarrow\; i = r \;\;\land \;\;  \{ j,s  \} \subseteq D \;\;\;\; \text{for} \;\;\;\; 0  \leq i,r <p 
 \end{equation}
The folding is extended to  markings  by $\eqcl[D]{M} :=   
\{ \eqcl[bf(D)]{s}| s \in M \} $\footnote{We will show that in each reachable marking of the cycloids under investigation each class contains at most one token.}. 
If $D = \{0,\cdots,\beta-1\} $ the folding and the equivalence relation are called \emph{total} and denoted by 
$ \mathcal{C}_{bf}( \alpha, \beta, \gamma, \delta, \eqcl[]{M_0} ) $ and $\equiv_{bf}$, respectively. 

\end{definition}

The folding is defined on backward places modelling the channels of the cooperating processes. Therefore by the folding we switch from a message oriented synchronization to a shared variable synchronization mechanism.
If $j \in D$ then the process $a_{j\oplus_c 1}$ is sharing its backward input places.

  \begin{lemma} \label{eq-classes}
\begin{itemize}
       \item [a)] For $0 \leq i <p$  the class of $[s'_i,a_0]$ with respect to 
	$\equiv_{bf}$  is \\
       $S^{bf}_i :=\{[s'_i,a_0]\} \cup \{[s'_{i\oplus_p n},a_j]| 0 < j < \beta\}$ with $n= \alpha+\beta$.
       When used as an element we write $\eqcl[bf]{S_i}$ for this class.
       In particular, for $p=n$ this reduces to 
       $S^{bf}_i  = \{[s'_{i},a_j]| 0  \leq j < \beta\}$.
       In this case we obtain \\
       ${}^{\ndot}{\eqcl[bf]{S_i}} = \{ [t_i,a_j]|0  \leq j < c \}$
       and ${\eqcl[bf]{S_i}}^{\ndot} = \{ [t_{i\ominus_n 1},a_j]|0  \leq j < c \}$.
       \item [b)]  All classes of the relation $ \equiv_{bf(D)}$ contain only one element, with the exception of 
$S^D_i :=\begin{cases}
                              \{[s'_i,a_0]\} \cup \{[s'_{i\oplus_p n},a_j]| j \in D \setminus  \{ 0 \}  \}&  \, \text{if}  \, \, 0 \in D \\
                              \{[s'_{i\oplus_p n},a_j]| j \in D  \}&  \, \text{if}  \, \, 0 \notin D  
\end{cases}$
for $\;\; 0  \leq i < p$.\\
By $\eqcl[]{S^D_i}$ we denote the class with elements from $S^D_i$ with respect to $ \equiv_{bf(D)}$.
\end{itemize}     
\end{lemma} 

\begin{proof} 
 a) $[s'_i,a_0] \in S^{bf}_i$ by definition. The output transition of  $[s'_i,a_0]$ is by Lemma \ref{le-reg-coord} 
 $[s'_i,a_0]{}^{\ndot} = [t_{(i+n-1)\; mod  \; p},a_{\beta-1}]$. By Definition  \ref{de-bf-folding} the remaining elements of 
 $S^{bf}_i $ are $\prenbbw{[t_{(i+n-1)\; mod  \; p},a_{j }]} $ for $0 < j <\beta $. By Lemma  \ref{bw-input-place} b) we obtain 
 for second set of  $S^{bf}_i$ in this lemma : 
$ \{    \prenbbw{  [t_{(i+n-1)\; mod  \; p},a_{j } ] }| 0  \leq j < \beta -1 \} $
 $ = \{[s'_{(i+n-1+1)\; mod  \; p},a_{j+1 }]| 0  \leq j < \beta -1 \} $
 $ = \{[s'_{(i\oplus_p  n)},a_{j }]| 0  < j < \beta  \}$. From this follows also the particular case for $p=n $.\\
 b) This follows directly from a) as $ \equiv_{bf(D)}$ is equal or finer than $ \equiv_{bf}$.
 \qed\end{proof} 

\begin{figure}[htbp]
 \begin{center}
        \includegraphics [scale = 0.3]{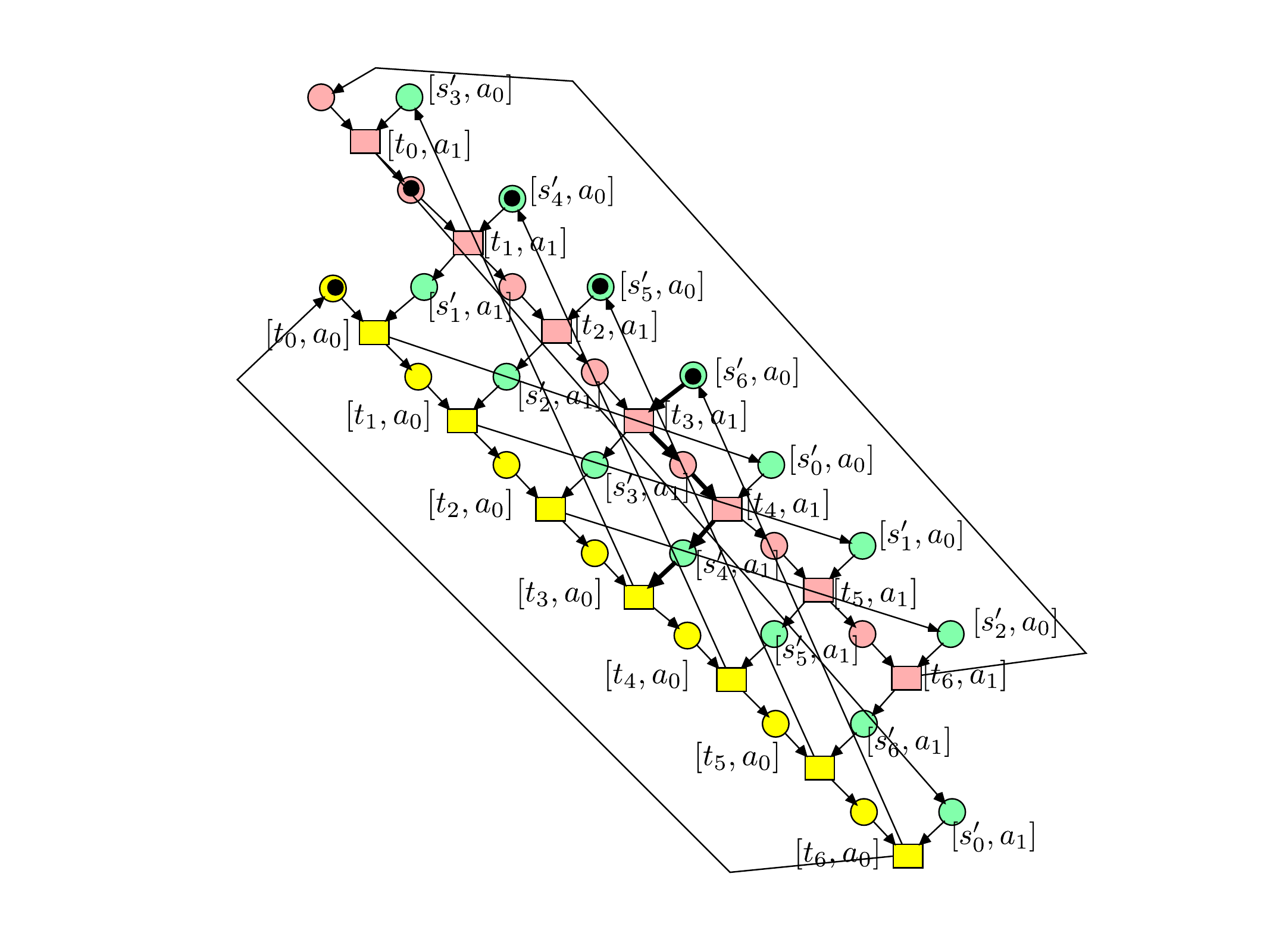}
        \caption{The regular cycloid system  $ \mathcal{C}( 3,2,1,4,M_0 ). $ }
        \label{3214}
      \end{center}
\end{figure}

For the cycloid  $ \mathcal{C}( 3,2,1,4 ) $ from Fig. \ref{3214} we obtain $A = 14, p =  \frac{A}{\beta} = 7, n = \alpha + \beta = 5$ and 
 $S^{bf}_0 = \{[s'_0,a_0],[s'_5,a_1]\}$,   $S^{bf}_1:= \{[s'_1,a_0],[s'_6,a_1]\} , \cdots ,  S^{bf}_6 := \{[s'_6,a_0],[s'_4,a_1]\}$.                      
For the cycloid  $ \mathcal{C}( 4,3,3,3) $ from Figure \ref{c-4-3-3-3}b) we obtain $A = 21, p = n = 7$ and 
 $S^{bf}_0 = \{[s'_0,a_0],[s'_0,a_1],$
 $[s'_0,a_2]\}$,   $S^{bf}_1 := \{[s'_1,a_0],[s'_1,a_1],$
 $[s'_1,a_2]\} , \cdots ,  S^{bf}_6 = \{[s'_6,a_0],[s'_6,a_1],[s'_6,a_2]\}$.   
For the definition of place invariants the following Lemma will be used.

\begin{lemma} \label{bf-path}
The places of a class $S^{bf}_i$ of a backward folding are included in a path, called \emph{bf-path} or 
$i$-\emph{bf-path},
as  follows: 
 $[s'_i,a_0], [t_{(i+n-1)\; mod  \; p},a_{\beta - 1}], 
 \cdots ,$ \\
$ [t_{(i+n-1)\; mod  \; p},a_{j}], [s_{(i+n-1)\; mod  \; p},a_{j}]^*,$
 $[t_{(i+n)\; mod  \; p},a_{j}], [s'_{(i+n)\; mod  \; p},a_{j}],$\\
 $[t_{(i+n-1)\; mod  \; p},a_{j-1}],$
 $\cdots $,
$ [t_{(i+n-1)\; mod  \; p},a_{0}]$.
 A $i$-\emph{bf-path} and a $a_j$-process share the place $ [s_{(i+n-1)\; mod  \; p},a_{j}]$ 
 $(0  \leq i < p,\; 1  \leq j < \beta)$.
For $p-n+1  \leq i < p-\alpha$ which is $i = p-n+1 +k \;\; (0 \leq k < \beta-1)$ the bf-path starting in $[s'_i,a_0] $ contains exactly one token in the regular initial marking.
\end{lemma} 

\begin{proof} 
 By  Lemma \ref{le-reg-coord} the output transition of $[s'_i,a_0]$ is 
 $[s'_i,a_0]{}^{\ndot} = $ 
 \newline $[t_{(i+n-1)\; mod  \; p},a_{\beta-1}]$. 
By induction, from $[t_{(i+n-1)\; mod  \; p},a_{j}] \;\; (0<j<\beta)$ we come in 4 steps to 
 $[t_{(i+n-1)\; mod  \; p},a_{j-1}]$
 which gives for $j = 1$ the end of the bf-path $[t_{(i+n-1)\; mod  \; p},a_{0}]$.
 The place shared with the $a_j$-process  is marked in the lemma by an asterisk.
 The place $[s_k,a_{k+1}] \;\; (0 \leq k < \beta-1)$ contains the single token the bf-path starting in $[s'_i,a_0] $.
 \qed\end{proof}  
 
 For an example of a bf-path, see the highlighted path from $[s'_6,a_0]$ to $[t_3,a_0]$ in Figure \ref{3214} with $n=5, p = 7$. 
  It is sharing the place
 $[s_3,a_1]$  with the $a_1$-process.
The place $[s_0,a_{1}] $ contains the single token of the bf-path starting in $[s'_3,a_0] $. 

It is important to prove that under a mild restriction the backward folding of a cycloid is 
safe and live.

\begin{theorem} \label{bf-safe}
 The  backward folding $ \mathcal{C}_{bf(D)}( \alpha, \beta, \gamma, \delta, \eqcl[]{M_0}  ) $ of a regular cycloid system
 $ \mathcal{C}( \alpha, \beta, \gamma, \delta, M_0 ) $ with regular $M_0$ and $ n - 1 \leq p$ is a safe net, i.e. in each reachable marking each place contains at most one token.
\end{theorem} 

\begin{proof} 
Since the cycloid system 
$ \mathcal{C}( \alpha, \beta, \gamma, \delta, M_0 ) $ is safe (Theorem 5.4 of \cite{Valk-2019}), it is sufficient to prove that each  equivalence class  $S^{bf}_i $ (Lemma \ref{eq-classes}) is contained in a S-invariant containing exactly one token.
Thus the same is holding for the classes $S^D_i$ of $ \equiv_{bf(D)}$.
As all places of a cycloid have exactly one input and output transitions the S-invariants can be defined by cycles. We distinguish two cases, namely those where  $[s'_i,a_0]$ is marked in the regular initial marking 
 $p-\alpha  \leq i < p$ and the complementary case $0 \leq i < p - \alpha$.

   Case 1: $[s'_i,a_0]$ is marked. Starting from $[s'_i,a_0]$ the cycle to be constructed initially is the bf-path until
   the $a_0$-process is reached and then follows this process until the input-transition of $[s'_i,a_0]$ is reached.
    Formally, since $[s'_i,a_0]$ is marked, we have  $p-\alpha  \leq i < p$ by Lemma \ref{regular-M0-reg-coord} and $i$ can be represented as $i= p-\alpha +k$ with $0  \leq k < \alpha$. Then the input transition of $[s'_i,a_0]$ is $[t_i,a_0] = [t_{p-\alpha +k},a_0]$.  The cycle is starting in $[s'_{p-\alpha +k},a_0]$
   and follows  the bf-path  which ends by  Lemma \ref{bf-path} in
  $[t_{(i+n-1)\; mod  \; p},a_0] = [t_{(p-\alpha +k+\alpha + \beta-1)\; mod  \; p},a_0] = 
 [t_{(p +k + \beta-1)\; mod  \; p},a_0] =
  [t_{k + \beta-1},a_0]
 $. From this transition the cycle follows the $a_0$-process until the 
 input transition of $[s'_i,a_0]$ namely $[t_i,a_0] = [t_{p-\alpha +k},a_0]$ is reached. This is possible without passing the end of the $a_0$-process to
 $[s_{p-1},a_0]$ if the inequality $k+\beta -1 \leq p - \alpha +k $ is holding. The inequality
  follows from the condition $n-1  \leq p$ of the lemma by $n-1  \leq p  \Leftrightarrow  \alpha + \beta -1  \leq \alpha + p - \alpha \Leftrightarrow \beta -1  \leq p - \alpha 
 \Leftrightarrow \beta -1+k \leq p - \alpha +k$.
The cycle contains a token in  
 $[s'_{i},a_0] $. The remaining forward places  are unmarked since by Lemma \ref{bf-path} the 
  $i$-bf-path is sharing the place $ [s_{(i+n-1)\; mod  \; p},a_{j}]$ with the $a_j$-process, which contains a single token in the place  $ [s_{j-1},a_{j}]$. The places are different since 
  $ [s_{(i+n-1)\; mod  \; p},a_{j}] = [s_{(p-\alpha+k+\alpha+\beta-1)\; mod  \; p},a_{j}] =
  [s_{k+\beta-1},a_{j}]$ and
  $j < \beta \Rightarrow j-1 < \beta -1 +k$ and $\beta +k < \beta + \alpha = n < p$ by the assumption of the 
  theorem.
 Also the places of the $a_0$-process are unmarked since the place $[s_{p-1},a_0]$ is not contained,
 
 Case 2: $[s'_i,a_0]$ is unmarked, hence 
 $0 \leq i < p - \alpha$. Here we consider two subcases: 
 $0 \leq i < p-n+1$ and $p-n+1  \leq i < p-\alpha$. 
 
 Case 2.1: $0  \leq i < p-n+1$. Again by Lemma \ref{bf-path} the bf-path starting in 
 $[s'_i,a_0]$ ends in $[t_{(i+n-1)\; mod  \; p},a_0]$ with $0 \leq i  \leq p-n$.
 Here, contrary to case 1, we are passing the marked place $[s_{p-1},a_0]$ to reach the input transition 
 $[t_i,a_0]$ of $[s'_i,a_0]$.
  The remaining places of the cycle are unmarked by the same arguments as before.

%
 Case 2.2: $p-n+1  \leq i < p-\alpha$  or $i= p-n+1+k$ with $0  \leq k < \beta -2$. 
 By the last sentence in Lemma \ref{bf-path} the bf-path starting in $[s'_i,a_0]$ contains exactly one token in the regular initial marking. This bf-path is extended to a cycle as in Case 1.                                                                                                                                                                                      
  \qed\end{proof}

\begin{figure}[htbp]
 \begin{center}
        \includegraphics [scale = 0.313]{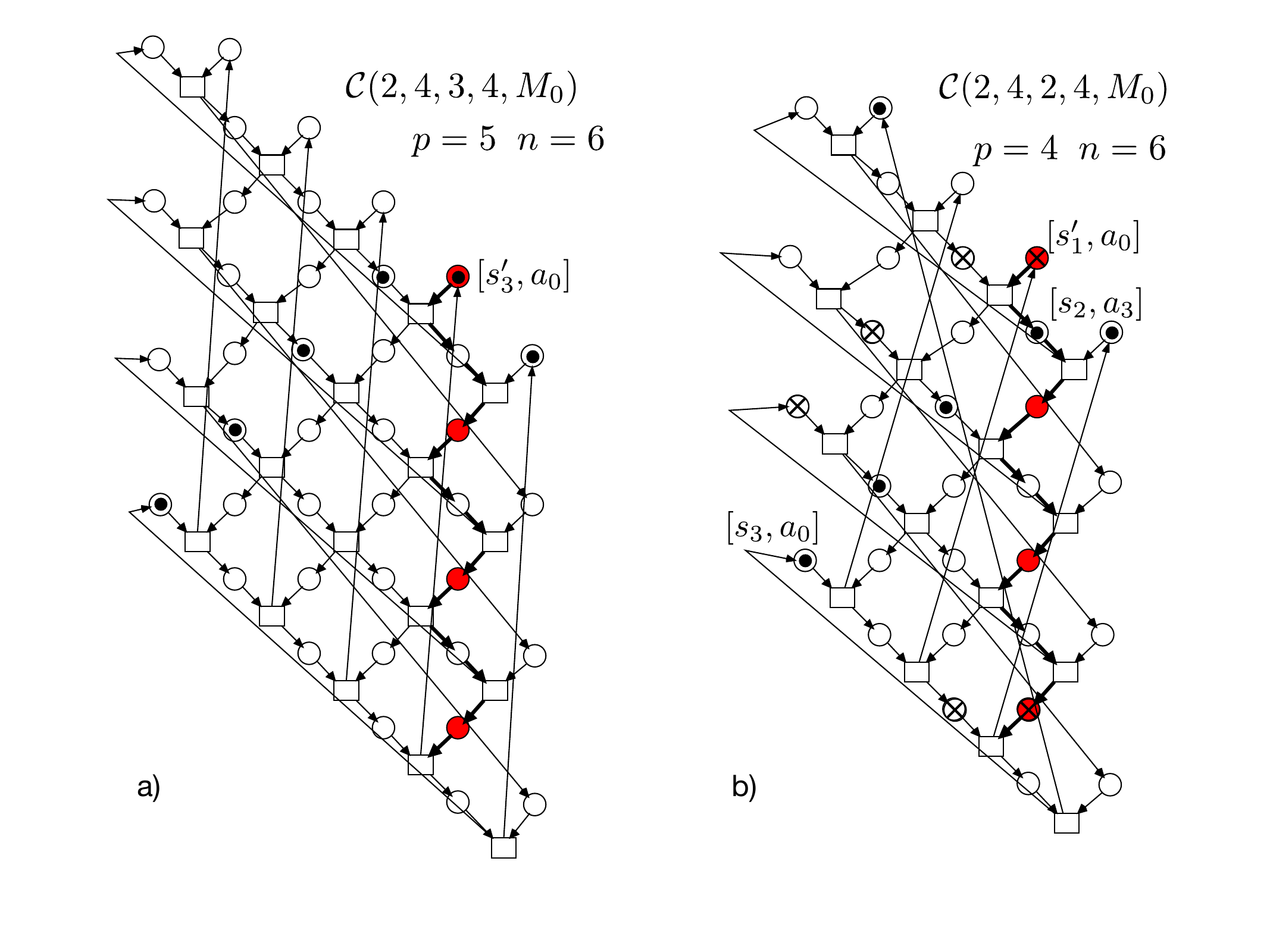}
        \caption{Cycloids with $n-1 = p$ and $n-2 = p$}
        \label{n-1=p}
      \end{center}
\end{figure}

By the cycloid systems of Figure \ref{n-1=p} the proof of Theorem \ref{n-1=p} is illustrated and it is shown that the assumption $ n - 1 \leq p$ of the theorem is sharp. 
For the cycloid system in a) the equivalence class   $S^{bf}_3 $    is given by places  represented by grey (or red) filled circles. They are connected by the $3$-bf-path in bold arrows, which is contained in a single marked cycle. This does not hold for the $1$-bf-path in part b) of the figure. A cycle contains two tokens in $[s_3,a_0]$ and $[s_2,a_3]$ in the given regular initial marking. The backward folding $ \mathcal{C}_{bf}(2,4,2,4,  \eqcl[]{M_0}  ) $ is not safe since the places marked by a cross represent a reachable marking where $S^{bf}_1$ contains two tokens.

\begin{theorem} \label{bf-live}
 Let be  $ \mathcal{C}_{bf(D)} =\mathcal{C}_{bf(D)}( \alpha, \beta, \gamma, \delta,  \eqcl[D]{M_0}  ) $ the  backward folding of a regular cycloid system
 $\mathcal{C} =  \mathcal{C}( \alpha, \beta, \gamma, \delta, M_0 ) = (S,T,F,M_0) $ with $ n - 1 \leq p$.
 \begin{itemize}
       \item [a)]   
 For each transition $t\in T$ and all markings $M, M'$ of $\mathcal{C} $
 $$	 M \stackrel{t}{\rightarrow}M'  \;\; \Leftrightarrow \;\; \eqcl[D]{M}  \xrightarrow[\text{\tiny{bf(D)}}]{\text{t}} \eqcl[D]{M'}
 $$
 where $\stackrel{t}{\rightarrow}$ and $ \xrightarrow[\text{\tiny{bf(D)}}]{\text{t}}$ are the transition relation of the net 
 $\mathcal{C}$ and $ \mathcal{C}_{bf(D)} $, respectively. 
 \item [b)]  $ \mathcal{C}_{bf(D)}$ is live.
\end{itemize}   
 \end{theorem} 

\begin{proof} 
a) We first prove the equivalence of the activations $M \stackrel{t}{\rightarrow} \;\; \Leftrightarrow \;\; \eqcl[D]{M}  \xrightarrow[\text{\tiny{bf(D)}}]{\text{t}} $.
\\
Let be $t = [t_i,a_j]$ a transition of $ \mathcal{C}$. If $M \stackrel{t}{\rightarrow} $ then the input places 
$ [s_{i\ominus_p 1},a_j]$  and  $[s'_{h  },a_{j\oplus_p  1}]$ (see Figure \ref{proof-live}) are marked. It follows that the input places $ [s_{i\ominus_p 1},a_j]$  and  $\eqcl[D]{ [s'_{h  },a_{j\oplus_p  1}]}$ of $[t_i,a_j]$ in $ \mathcal{C}_{bf(D)}$ are marked and $\eqcl[D]{M}  \xrightarrow[\text{\tiny{bf(D)}}]{\text{t}} $.  
\begin{figure}[htbp]
 \begin{center}
        \includegraphics [scale = 0.2]{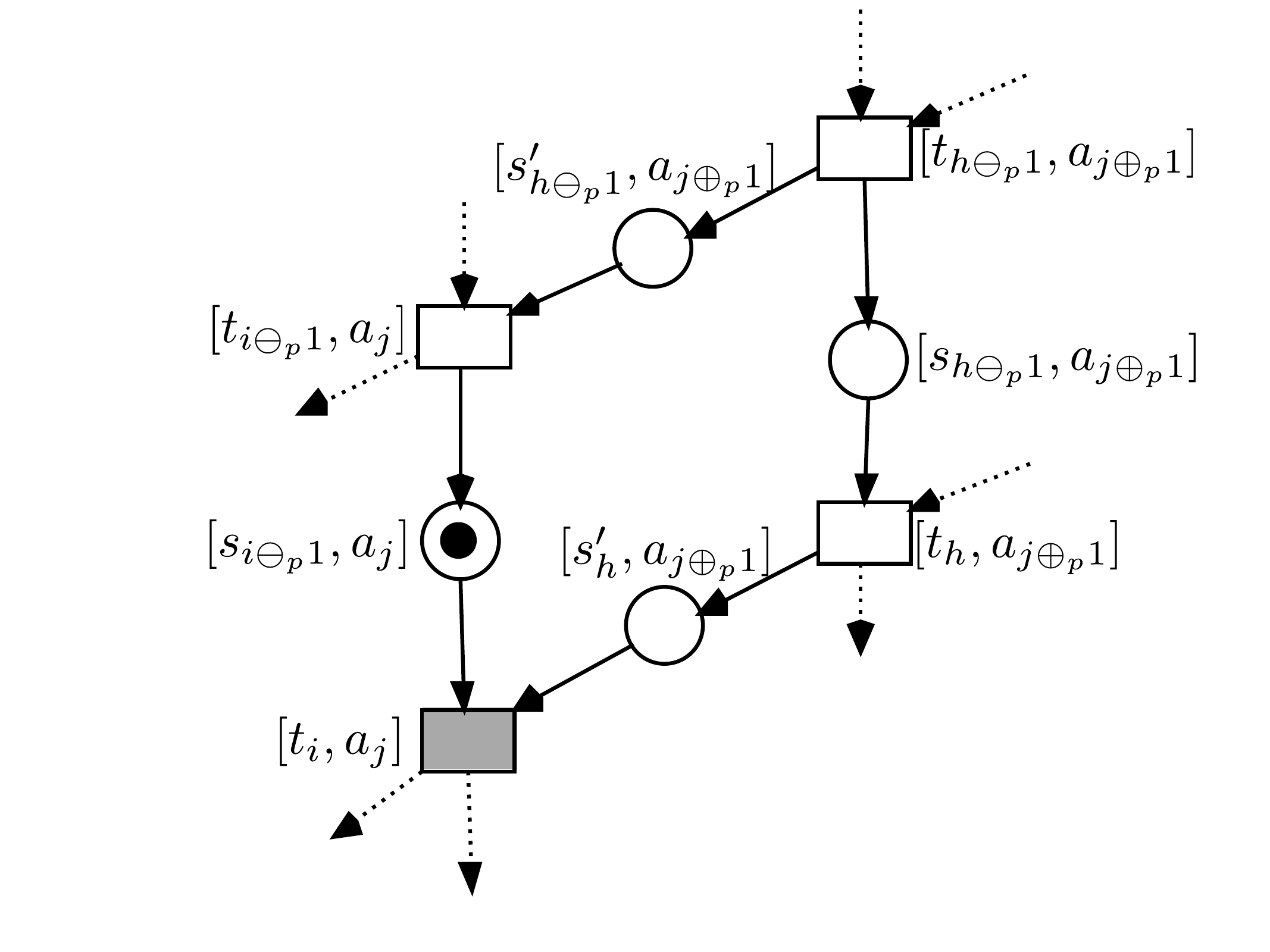}
        \caption{Case 2 of the proof of Theorem \ref{bf-live}}
        \label{proof-live}
      \end{center}
      
\end{figure}
The inverse implication is proved by contradiction: $\lnot \; M \stackrel{t}{\rightarrow} \;\; \Rightarrow \;\; \lnot \; \eqcl[D]{M} 
\xrightarrow[\text{\tiny{bf(D)}}]{\text{t}} $. If $\lnot \; M \stackrel{t}{\rightarrow}$ there are two cases:\\
Case 1: $ [s_{i\ominus_p 1},a_j] \notin M$. As $ [s_{i\ominus_p 1},a_j] $ is unchanged in $ \mathcal{C}_{bf(D)}$ also $\lnot \; \eqcl[D]{M}  \xrightarrow[\text{\tiny{bf(D)}}]{\text{t}} $.\\
Case 2: $ [s_{i\ominus_p 1},a_j] \in M$ but $[s'_{h  },a_{j\oplus_p  1}] \notin M$. To prove 
$\lnot \; \eqcl[D]{M}  \xrightarrow[\text{\tiny{bf(D)}}]{\text{t}} $ we deduce that the class 
$\eqcl[bf(D)]{[s'_{h  },a_{j\oplus_p  1}]}$  
of $[s'_{h  },a_{j\oplus_p  1}]$ is not in $\eqcl[]{M_0} _D$.
This is the case if all elements of the class 
 $S^{bf}_h =\{[s'_h,a_0]\} \cup \{[s'_{h\oplus_p n},a_j]| 0 < j < \beta\}$  with $0 \leq i < p$ 
are unmarked. For the element $[s'_{h  },a_{j\oplus_p  1}] $ of this set, this property is part of the assumption of Case 2. To prove it for the remaining elements of this set consider the bf-path of  $\eqcl[D]{s'_h}$, as defined in Lemma \ref{bf-path} containing all these elements. As shown in Lemma  \ref{bf-path} this bf-path is extended to a S-invariant. We define a rerouting of the corresponding cycle by replacing the sub-path \\
$ [t_{h\ominus_p  1},a_{j\oplus_p  1}], [s_{h\ominus_p  1},a_{j\oplus_p  1}],[t_{h  },a_{j\oplus_p  1}],  [s'_{h  },a_{j\oplus_p  1}], [t_i,a_j]$
by \\
$ [t_{h\ominus_p  1},a_{j\oplus_p  1}], [s'_{h\ominus_p  1},a_{j\oplus_p  1}],[t_{i\ominus_p  1},a_j],  
[s_{i\ominus_p 1},a_j], [t_i,a_j]$ 
(see Figure \ref{proof-live}), which is also a S-invariant containing the place $ [s_{i\ominus_p 1},a_j] $ which is marked in Case 2. Since the net is safe by Theorem \ref{bf-safe} all places different to  $ [s_{i\ominus_p 1},a_j] $  are unmarked, including the elements of $S^{bf}_h $.\\
b) By part a) of the proof $ \mathcal{C}_{bf(D)}$ is behavioural equivalent to $ \mathcal{C}$ which is live by Theorem 5.4 of \cite{Valk-2019}.
\qed\end{proof}

When in a circular traffic queue $tq(c,g)$ a car is removed the system should behave like a circular traffic queue $tq(c-1,g+1)$ with one car less and one gap more. 

\begin{lemma} \label{reg-proc}
a) For a given regular cycloid  $ C =\mathcal{C}( \alpha, \beta, \gamma, \delta, M_0 ) $ with $\beta > 1$
(and process length $p$)
 the cycloid 
$ C' =\mathcal{C}( \alpha +1 , \beta -1, p - (\alpha + 1), \beta-1, M' ) $ 
has the same process-length 
as $C$.\\
b)    Each regular cycloid $ C'' =\mathcal{C}( \alpha +1 , \beta -1, \gamma'', \delta'', M'') $\footnote{$M_0,M',M''$ are the corresponding regular initial markings.} with the same process-length $p$ is isomorphic to $C'$.
\end{lemma}

\begin{proof}  
a) If $A'$ is the area of $C'$ the process-length of $C'$ is \\
$p' = \frac{A'}{\beta-1}= \frac{1}{\beta-1}\cdot ((\alpha+1)\cdot(\beta-1) + (\beta-1) \cdot (p - (\alpha +1) ) =p$.\\
b) To prove that $C'$ and $C''$ are isomorphic, by Theorem \ref{symmetry} b) it is sufficient to show that
for some $ q \in \Natp $ we have
\begin{equation}\label{Gl-reg-proc}
\gamma'' =  \gamma' - q \cdot (\alpha+1)  \text{\;\;and\;\;}  \delta'' = (\beta-1)+ q \cdot (\beta-1)  \text{\;\;for\;\;} 
\gamma' = p-(\alpha+1)
\end{equation}
Since $C''$ is supposed to be regular we have $\delta'' = r \cdot (\beta-1)$ for some $ r \in \Natp $.
$C''$ is also required to have the same process-length 
$p'' = \frac{1}{\beta-1}\cdot ((\alpha+1)\cdot \delta'' + (\beta-1) \cdot \gamma'' ) = 
\frac{1}{\beta-1}\cdot ((\alpha+1)\cdot r \cdot (\beta-1) + (\beta-1) \cdot \gamma'' ) =$ 
$ (\alpha+1)\cdot r  +  \gamma''  $
as $p = \gamma' + (\alpha +1) $. 
$p=p''$ gives 
$ \gamma' + (\alpha +1) =  (\alpha+1)\cdot r + \gamma''  $
and
$ \gamma''   = \gamma' + (\alpha +1) - (\alpha+1)\cdot r = \gamma' - (r-1) \cdot (\alpha+1)$,
 hence $q = r-1$
 as required in Equation (\ref{Gl-reg-proc}).
The observation $\delta'' = r \cdot (\beta-1) = (q+1) \cdot (\beta-1) = (\beta-1)+ q \cdot (\beta-1)$
satisfies the second part of Equation (\ref{Gl-reg-proc}).
\qed\end{proof}

A bf-folding with $D =  \{ 0,\beta-1 \} $ is merging the backward input places of the processes of the traffic items $a_{\beta-1}$ and $a_{\beta-2}$. Then the $a_{\beta-1}$-process can be stopped or eliminated such that the remaining cycloid is still live and safe with the same process length. This is done with the cycloid 
$ \mathcal{C}_0$ in Figure \ref{2346-3252}. Via the intermediate step of the bf-folding $ \mathcal{C}_1$ a cycloid $ \mathcal{C}_2$ is obtained, where the $a_2$-process is eliminated.

\begin{figure}[htbp]
\hspace{-1 cm}
        \includegraphics [scale = 0.2 ]{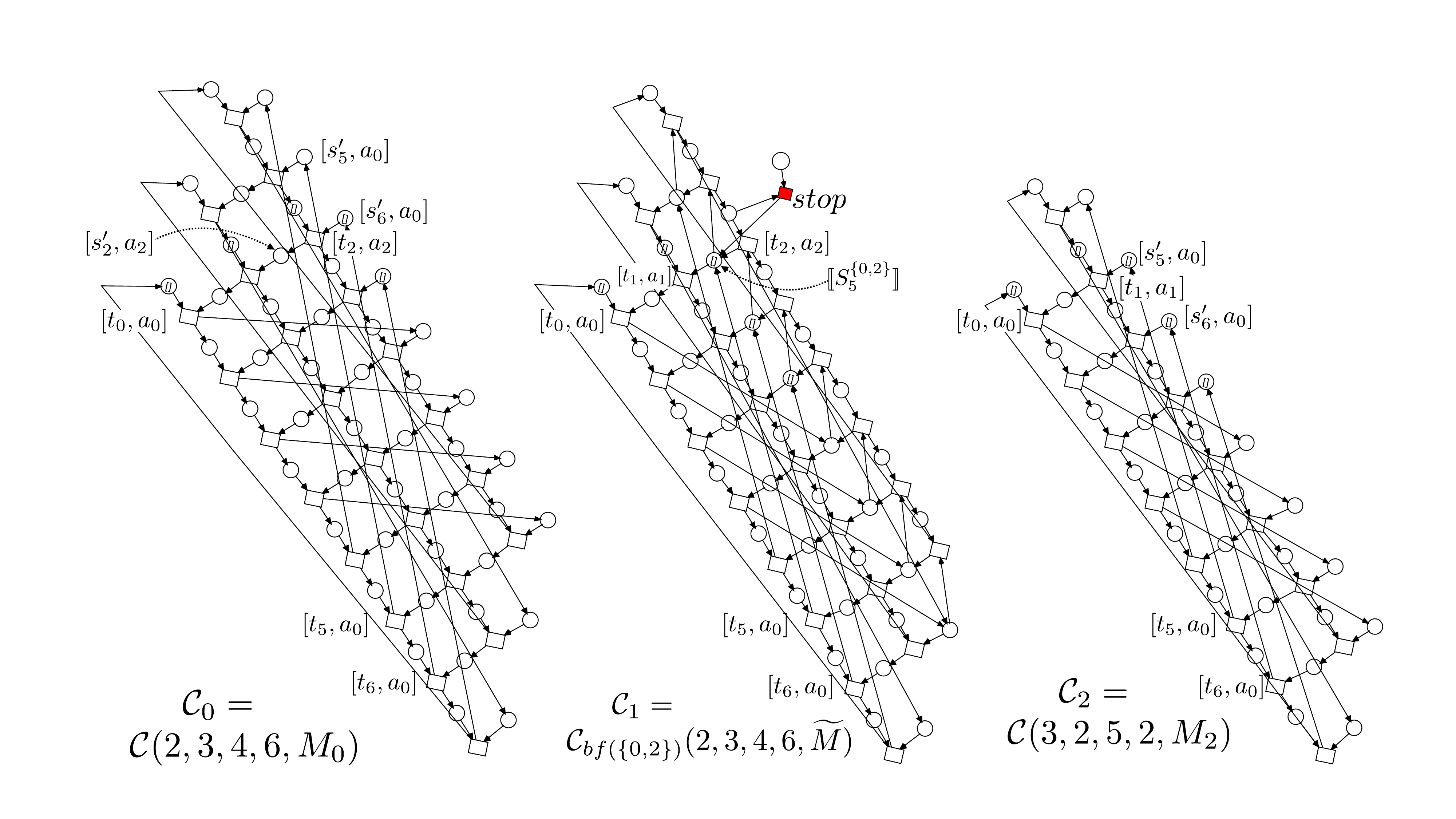}
        \caption{Two cycloids and bf-folding $\mathcal{C}_{bw(2,1)}( 2,3,4,6,\widetilde{M} )$ as 	intermediate step. }
        \label{2346-3252}
\end{figure}

\begin{definition} \label{M_downarrow}
For a bf-folding $ \mathcal{C}_1 =\mathcal{C}_{bf(D)}( \alpha, \beta, \gamma, \delta, \eqcl[D]{M_0} ) $ 
with \\
$D= \{ 0,\beta-1 \} $ we define 
$[\mathcal{C}_1]\downarrow(\beta-1)$ by deleting all transitions and forward places of the 
$a_{\beta-1}$-process.
\end{definition}

\begin{theorem} \label{cyc-beta-1}
Let be   $ \mathcal{C}_0 =\mathcal{C}( \alpha, \beta, \gamma, \delta,M_0) $ a regular cycloid 
with $\beta  > 1$.\\
a) $ \mathcal{C}_2 = \mathcal{C}( \alpha+1, \beta-1,  p - (\alpha + 1), \beta-1,M_2 )$\footnote{$M_0,M_2$ are the corresponding regular initial markings.} 
is behaviour equivalent to 
     $\mathcal{C}_1 = \mathcal{C}_{bf(\{0,\beta-1)\}}( \alpha, \beta, \gamma, \delta,\widetilde{M} ) $
     where 
     $\widetilde{M} = \eqcl[]{M_0} -\{[s_{\beta \ominus_p  2},a_{\beta-1}]\} \cup 
\{\eqcl[ ]{S^ {\{0,\beta-1\}}_{p-(\alpha+1)}}\}$
b) $\mathcal{C}_2$ is isomorphic to $\mathcal{C}_3 =[C_1]\downarrow(\beta-1)$.
\end{theorem}

\begin{proof} 
a) By the elimination of the token in $[s_{\beta \ominus_p  2},a_{\beta-1}]\in M$ all transitions of the $a_{\beta - 1}$-process are dead. Therefore it is sufficient to prove part b) of the theorem.\\
b) We denote by $f_D$ the folding from  $ \mathcal{C}_0$ to  $ \mathcal{C}_1$ defined by $ \equiv_{bf(D)}$
and define an isomorphism $\phi$ from  $ \mathcal{C}_3$ to  $ \mathcal{C}_2$. 
Since $D = \{ 0,\beta-1 \}$ the folding is defined for places or transitions $x$ of  $ \mathcal{C}_0$ by 
$f_D(x) = 
\begin{cases}
                      \eqcl[ ]{S^D_{i}}&  \, \text{if}  \, \, x \in S^D_{i}   \\
                      x & \, \text{otherwise}
\end{cases}$, where 
$S^D_{i}:=\{[s'_i,a_0]\} \cup \{[s'_{i\oplus_p n},a_{\beta-1}]\}$  with $0 \leq i < p$ (Lemma \ref{eq-classes} b).
Now we define the isomorphism by  
$\phi(x) = 
\begin{cases}
                      [s'_i,a_0] &  \, \text{if}  \, \, x =  \eqcl[ ]{S^D_{i}}  \\
                      x & \, \text{otherwise}
\end{cases}$, both for $0  \leq i <p$. \\
The isomorphism is a renaming of $\eqcl[ ]{S^D_{i}} $ in  $ \mathcal{C}_3$ to $[s'_i,a_0] $ in  $ \mathcal{C}_2$.
All different names of places or transitions are unchanged. Therefore, to prove that $\phi$ is an isomorphism, it is sufficient to show that the input and output transitions of $\eqcl[ ]{S^D_{i}} $ in  $ \mathcal{C}_3$ and
 $[s'_i,a_0] $ in  $ \mathcal{C}_2$ are equally named.
 The input and output transitions of $\eqcl[ ]{S^D_{i}} $ are those of $[s'_i,a_0]$ and $[s'_{i\oplus_p n},a_{\beta-1}]$ in  $ \mathcal{C}_0$. Using Lemma \ref{le-reg-coord} a) we obtain
$[s'_{i},a_{0}]^{\ndot} = [t_{(i+n-1)\; mod  \; p},a_{\beta-1}]$
and by 
Lemma \ref{le-reg-coord} b) we obtain
$[s'_{i\oplus_p n},a_{\beta-1}]^{\ndot} = [t_{(i+n-1)\; mod  \; p},a_{\beta-2}]$. The input places have the same indices as their output places. They are all summarized in Table \ref{input-output}.
%
\begin{table}[htbb]
 \begin{center}
\caption{}
\label{input-output}
 \begin{tabular}{|c|c|c|c|c}
  \hline 
$\mathcal{C}_0  $             &$\mathcal{C}_1$           & input transition &  output transition              \\  \hline\hline
$[s'_i,a_0]$    &  \multirow{2}{*}{$S^D_i$}& $[t_i,a_0]$ &\cellcolor{Gray}$[t_{(i+n-1)\; mod  \; p},a_{\beta-1}]$     \\ 

$[s'_{i\oplus_p n},a_{\beta-1}]$     & & \cellcolor{Gray}$[t_{i\oplus_p n},a_{\beta-1}]$&$[t_{(i+n-1)\; mod  \; p},a_{\beta-2}]$       \\  \hline
       \end{tabular}  
       \end{center}
       \end{table}

By the definition of  $ \mathcal{C}_2$ the transitions of the $\beta-1$-processes are eliminated 
(marked with grey background in Table \ref{input-output}). The parameters of  $ \mathcal{C}_2$ in relation to  $ \mathcal{C}_1$ are $\alpha' =\alpha+1, \beta'= \beta-1, p'= p, n' = n$ and $ i' = i$.
We have to prove that the input and output transitions of $\phi( \eqcl[ ]{S^D_{i}}) =  [s'_i,a_0]$ are preserved by the isomorphism. This is holding since $\phi( {}^{\ndot}{\eqcl[ ]{S^D_{i}}}) = \phi([t_{i},a_{0}]) = [t_{i'},a_{0}]= [t_{i},a_{0}]$
and $\phi( \eqcl[ ]{S^D_{i}} ^{\ndot}) = \phi([t_{(i+n-1)\; mod  \; p},a_{\beta-2}]) =
 [t_{(i'+n'-1)\; mod  \; p'},a_{\beta'-1}] $
 since $\beta-2=\beta'-1$.

  It remains to prove that $\phi(\widetilde{M}) = M_2 $ is the regular initial marking of $C_2$. From Corollary \ref{regular-M0-reg-coord} we obtain 
  $M_0 = \{[s_{p-1},a_0]\} \cup \{ [s_i,a_{i+1}] | 0\leq i < \beta-1  \}   \cup $\\$
 \{ [s'_i,a_0] | p-\alpha\leq i < p\}$. 
 By $\widetilde{M} = M_0 -\{[s_{\beta \ominus_p  2},a_{\beta-1}]\} \cup 
\{\eqcl[ ]{S^ {\{0,\beta-1\}}_{p-(\alpha+1)}}\}$
the forward token of the 
 $a_{\beta-1}$-process is deleted, which is in agreement with $\beta' = \beta -1$.
 Due to $\alpha' = \alpha+1$ the token in the backward place $[s'_{p-\alpha},a_0]$ is missing. It is added by
 $\phi(\eqcl[ ]{S^ {\{0,\beta-1\}}_{p-(\alpha+1)}}) = [s'_{p-(\alpha+1)},a_0]= [s'_{p-\alpha'},a_0]$.
\qed\end{proof}

Figure \ref{2346-3252} shows the cycloid $  C_0 =\mathcal{C}( \alpha, \beta, \gamma, \delta,M_0) = \mathcal{C}( 2,3,4,6,M_0 ) $ and its folding $\mathcal{C}_{bf(\{0,2\})}( 2,3,4,6,\widetilde{M} )$ as 	intermediate step to the cycloid $ C_2 = \mathcal{C}( \alpha+1, \beta-1, p-(\alpha+1), \beta-1,M_2 )=   \mathcal{C}( 3,2,5,2) $. The transition $stop$ in $\mathcal{C}_1 $ simulates the transit from the marking $M_0$ to $\widetilde{M}$. The mappings in the proof for $i = 5$   and $D= \{0,2\}$ with $S^D_{5}:=\{[s'_5,a_0], [s'_{2},a_{2}]\}$ are
$f_D( [s'_5,a_0]) = f_D(  [s'_{2},a_{2}]) = \eqcl[ ]{S^D_{5}} $  and $\phi_D( \eqcl[ ]{S^D_{5}}) = [s'_{5},a_{0}]$.

\begin{definition} \label{def-stop} 
Given a total backward folding $ \mathcal{C}_{bf}( g,c,c,c, \eqcl[]{M_0} ) $
(Definition \ref{de-bf-folding}),
a \emph{stop resilient cycloid}
  $ \mathcal{C}^{stop}_{bf}( g,c,\eqcl[]{M_0} ) $ is defined by adding 
       transitions $[t_{stop},a_{j}], (0  \leq j < c)$ with
       ${}^{\ndot}{[t_{stop},a_{j}]}= \{[s_{j\ominus_n 1},a_j]\}$,
       ${[t_{stop},a_{j}]}^{\ndot} = \{\postnbbw{[t_{j},a_{j}]}\}$.
       Then the process of $a_{j}$ is called a \emph{stoppable} process.
\end{definition}

A transition $[t_{stop},a_j]$ deletes the control token in $[s_{j\ominus_n 1},a_j]$ of the $a_j$-process, but in same time it generates a permit signal by a token in the backward output place $\{\postnbbw{[t_{j},a_{j}]}\}$, just as the alternative transition $[t_{j},a_{j}]$
 would have done in its next step.
In terms of cars and gaps in a lane, this means that the stopping car is still able to leave the lane  giving place to the next car. 
We now formulate the final result of this section.


\begin{theorem} \label{th-stop}
Let be $s \in \Natp$ with $0 < s <c$.
\begin{itemize}
   \item [a)]   
          Stopping any number of $s$ processes   of  $  \mathcal{C}_0 = \mathcal{C}^{stop}_{bf}( g,c,\eqcl[]{M_0} ) $ 
          (by transitions $[t_{stop},a_{j}]$) 
          results in a safe net, where    	all transitions $[t_i, a_j]$ from
           not stopped processes are live. 
   \item [b)] Deleting all transitions $[t_i, a_j]$ and forward places $[s_i, a_j]$ of stopped processes and the stop 	transitions  $[t_{stop},a_{j}]$ in  $ \mathcal{C}_0$ results in a folding
  $\mathcal{C}_1 = \mathcal{C}_{bf}( g+s,c-s,\eqcl[]{M_1} ) $, where the underlying nets are isomorphic
  with respect to initial markings $\eqcl[]{M'_0}$ of  $ \mathcal{C}_0$ and $\eqcl[]{M'_1}$  of  $ \mathcal{C}_1$,
  respectively. $\eqcl[]{M'_0}$ is a regular follower marking of $\eqcl[]{M_0}$ in  $ \mathcal{C}_0$ and 
  $\eqcl[]{M'_1}$ is a regular follower marking of $\eqcl[]{M_1}$ in  $ \mathcal{C}_1$.
\end{itemize}   
\end{theorem} 
\begin{figure}[htbp]

\hspace{-1,6 cm}
        \includegraphics [scale = 0.22]{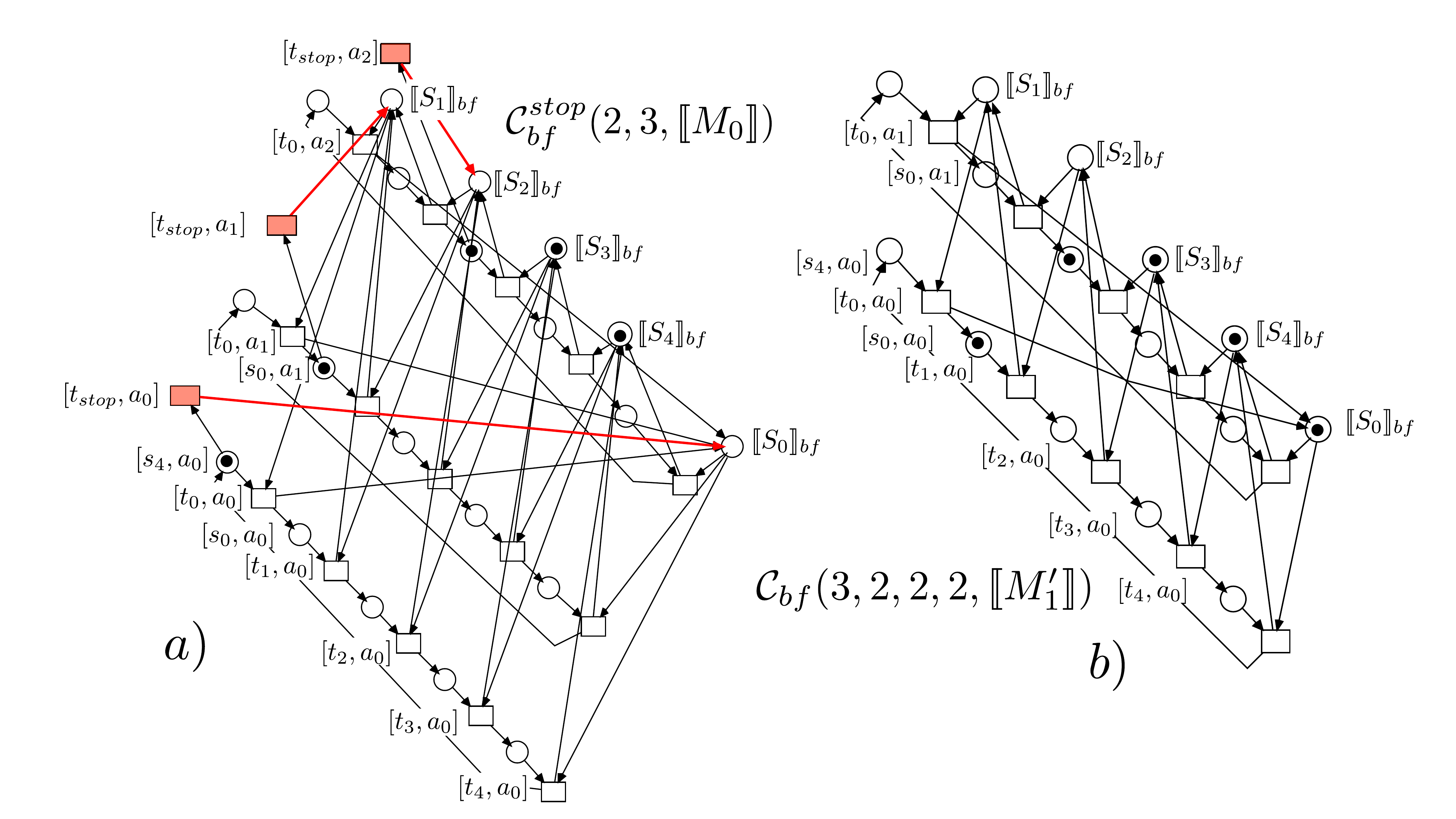}
        \caption{The stop resilient cycloid $ \mathcal{C}^{stop}_{bf}( 2,3 ) $ and the backward folding 
        $ \mathcal{C}^{}_{bf}( 3,2,2,2 ) $ }.
        \label{C-stop}
   
\end{figure}

\begin{proof} 
a) This follows from b) below since by Theorem \ref{bf-live} 
$ \mathcal{C}_1$ is live.\\
b) We  prove this part for a single stopping process $a_k$ (hence $s=1$), from which follows this part of the theorem  by induction.
 For a stopping process $a_{c-1}$ the proof is already done by Theorem \ref{cyc-beta-1},
since for $p=n$ we obtain $p-(\alpha+1)=p-(g+1) = (g+c)-(g+1) = c-1$. The isomorphism is the identity map.\\
b) If the stopping process is $a_k$ with $0 \leq k < c-1$ we define the isomorphism by 
$\phi([t_i,a_j]) = 
	\begin{cases}
		[t_i,a_{j-1}]&  \, \text{if}  \, \, j > k \\
		[t_i,a_j] & \, \text{otherwise}
\end{cases}$,
$\phi([s_i,a_j]) = 
	\begin{cases}
		[s_i,a_{j-1}]&  \, \text{if}  \, \, j > k \\
		[s_i,a_j]) & \, \text{otherwise}
\end{cases}$ 
and $\phi(\eqcl[bf]{S_i}) = \eqcl[bf]{S_i}$ for $0 \leq i < n$,
since by deleting the $a_k$-process the processes with higher indices reduce their indices.
After the occurrence of the transition $[t_{stop},a_k]$ (in this proof, not in a simulation), all its places and transitions are deleted. Therefore only the synchronization by backward places between $a_{k+1}$ and 
$a_{k\ominus_c 1}$ are to be considered. By Lemma \ref{eq-classes} a) we obtain 
$\phi({[t_i,a_{k+1}])}^{\ndot} = {[t_i,a_{k}]}^{\ndot}  = \eqcl[bf]{S_i}$ and 
$\phi([t_{i\ominus_n 1},a_{k-1}])=[t_{i\ominus_n 1},a_{k-1}]  \in {\eqcl[bf]{S_i}}^{\ndot}$.
This satisfied the structure of  $ \mathcal{C}_1$. 
Due to the renaming of $[t_i,a_{j}]  \, \text{for}  \, \, j > k$ by $\phi$ the image of the initial regular marking is not a regular marking of  $ \mathcal{C}_1$.
Due to the definition of $[t_{stop},a_k]$ a token is inserted into  $\postnbbw{[t_{k},a_{j}]}$. Therefore the transitions $[t_i,a_{j}]  \, \text{for}  \, \, 0  \leq j <k  $ can occur successively, resulting in a regular follower marking $\eqcl[]{M'_1}$ of  $ \mathcal{C}_1$, which is the $1$-regular marking of  $ \mathcal{C}_1$, as defined in Corollary \ref{regular-M0-reg-coord} b).
Also by this corollary the  the regular initial marking of $ \mathcal{C}_1$ reachable again. This is important to note with respect to the application of the induction step, as mentioned at the beginning of the proof. Also before
the next induction step the stop transitions must be added again. To finish the proof for the isomorphism we note that from the regular initial marking $\eqcl[]{M_0} $ of  $ \mathcal{C}_0$ the 1-regular marking $\eqcl[]{M'_0} $ is reachable satisfying $\phi(\eqcl[]{M'_0}) = \eqcl[]{M'_1}$.
\qed\end{proof}

\begin{figure}[htbp]
 \begin{center}
        \includegraphics [scale = 0.22]{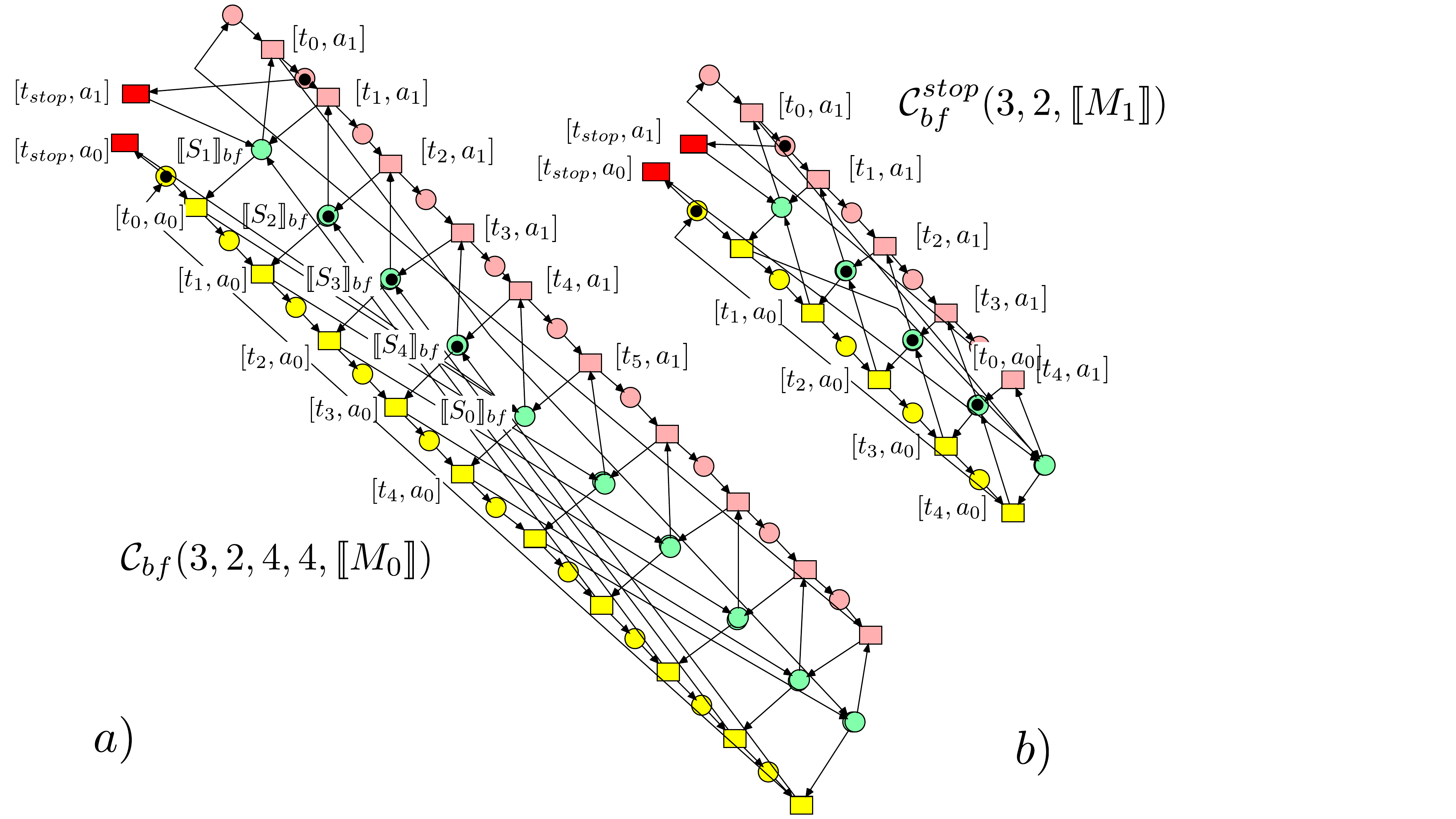}
        \caption{Counter-example to Theorem  \ref{th-stop} }
        \label{3244-vs-3222}
      \end{center}
\end{figure}
\vspace{-0.15  cm}
If possible by an arbitrary reached marking the stopping of several processes can occur in an interleaving mode. This property holds since a stopping transition 
$[t_{stop},a_j]$ generates the same token in a backward place as the transition $[t_{j},a_j]$ would do.
  

Figure \ref{C-stop} a) shows the folding  $ \mathcal{C}^{stop}_{bf}( 2,3,\eqcl[]{M_0} ) $. By the occurrence of one of the three stop-transitions and deleting the related transitions a folding is generated, which is isomorphic to the folding  $ \mathcal{C}_{bf}( 3,2,2,2,\eqcl[]{M'_1} ) $ in Figure \ref{C-stop} b). Then by reaching the regular initial marking  $\eqcl[]{M_2}$ and by adding the stopping transitions again we obtain 
$ \mathcal{C}^{stop}_{bf}( 3,2,\eqcl[]{M_2} )$. In the same way  by another induction step a folding 
$ \mathcal{C}^{stop}_{bf}( 4,1,\eqcl[]{M_3} )$  can be obtained.

Different to Theorem \ref{cyc-beta-1} the preceding Theorem \ref{th-stop} does not hold for process lengths greater than $n = c+g$. For the   backward folding 
$ \mathcal{C}_{bf}( 3,2,4,4, \eqcl[]{M_0}) $ of the regular cycloid  $ \mathcal{C}_{}( 3,2,4,4, M_0 ) $  in Figure \ref{3244-vs-3222}a) after the occurrence of the added transition $[t_{stop},a_0]$ and the transition sequence $[t_1,a_1],$ $[t_2,a_1],[t_3,a_1],[t_4,a_1] $ transition $[t_5,a_1]$ is not enabled. This is in contrast to $ \mathcal{C}^{stop}_{bf}( 3,2,\eqcl[]{M_1} ) $ in Figure \ref{3244-vs-3222} b)  where after the same transition sequence  $[t_0,a_1]$ is  enabled. Hence, stop resilience cannot implemented for the extensions like  
$ \mathcal{C}(g,c,\frac{g \cdot c}{\Delta},$
 $\frac{g \cdot c}{\Delta},M_0)$
in the way discussed here, but must be realized by iterating $ \mathcal{C}^{stop}_{bf}(g,c,\eqcl[]{M_0} ) $.


\section{Conclusion}\label{sec-conclusion}
Despite the tight synchronisation of sequential processes in the form of regular cycloids, it has been possible to extend the formalism in such a way that individual processes can fail without hindering the other processes.

\bibliographystyle{splncs03}
\bibliography{citations-rv}
\end{document}